\documentclass[english,a4paper,thm-restate,nolineno]{socg-lipics-v2021}

\title{The Localized Union-of-Balls Bifiltration}

\titlerunning{The Localized Union-of-Balls Bifiltration}

\author{Michael Kerber}{Institute of Geometry, Graz University of Technology, Austria}{kerber@tugraz.at}{ https://orcid.org/0000-0002-8030-9299}{Austrian Science Fund (FWF) grant P 33765-N}

\author{Matthias S\"ols}{Institute of Geometry, Graz University of Technology, Austria}{soels@tugraz.at}{}{Austrian Science Fund (FWF): grant W1230 }

\authorrunning{Kerber, S\"ols}

\Copyright{Michael Kerber,  Matthias S\"ols}



\supplement{A prototypical implementation can be found here: \url{https://bitbucket.org/mkerber/demo_absolute_2d/src/master/}.}

\funding{The authors acknowledge the support of the Austrian Science Fund (FWF).}%

\acknowledgements{The authors thank Anton Gfrerrer and Thomas Pock for helpful discussions.}

\hideLIPIcs  

\EventEditors{Erin W. Chambers and Joachim Gudmundsson}
\EventNoEds{2}
\EventLongTitle{39th International Symposium on Computational Geometry
	(SoCG 2023)}
\EventShortTitle{SoCG 2023}
\EventAcronym{SoCG}
\EventYear{2023}
\EventDate{June 12--15, 2023}
\EventLocation{Dallas, Texas, USA}
\EventLogo{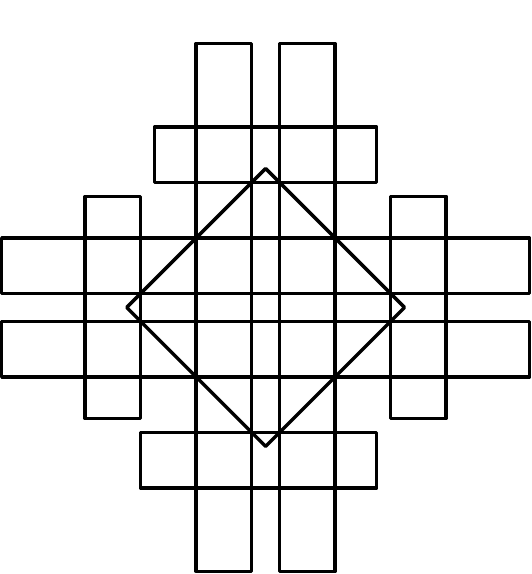}
\SeriesVolume{258}
\ArticleNo{46}

\usepackage[pdftex,dvipsnames,table]{xcolor} 

\usepackage{tikz}
\usetikzlibrary{arrows}

\usepackage{tabularx}
\usepackage{graphicx}
\usepackage{makecell}
\usepackage{booktabs}
\usepackage{extpfeil}

\usepackage{tikz-cd}

\usepackage[algo2e,vlined,ruled,linesnumbered]{algorithm2e} 
\SetKwIF{If}{ElseIf}{Else}{if}{}{else if}{else}{end if}%
\SetKwFor{While}{while}{}{end while}%
\SetKwFor{For}{for}{}{end for}%
\DontPrintSemicolon
\usepackage{algorithm}

\usepackage[nocompress,space]{cite}
\usepackage{hyperref}
\usepackage[noabbrev, nameinlink, capitalise]{cleveref}

\usepackage{lineno}

\usepackage{adjustbox}

\newcommand{\R}{\mathbb{R}}
\newcommand{\N}{\mathbb{N}}
\newcommand{\nrv}{\mathrm{Nrv}}
\newcommand{\Vor}{\mathrm{Vor}}
\newcommand{\eps}{\varepsilon}
\newcommand{\rb}{(L,L^{\mathfrak{C}}q)}

\usepackage{xspace}

\usepackage{mathrsfs}

\newcommand{\complex}{K}

\newcommand{\openball}[2]{B^o_{#1}{(#2)}}
\newcommand{\cover}{\mathcal{U}}
\newcommand{\ignore}[1]{}

\begin{document}

\maketitle

\begin{abstract} 
We propose an extension of the classical union-of-balls filtration
of persistent homology: fixing a point $q$, we focus our attention
to a ball centered at $q$ whose radius is controlled by a second
scale parameter. We discuss an absolute variant, where the union
is just restricted to the $q$-ball, and a relative variant where 
the homology of the $q$-ball relative to its boundary is considered.
Interestingly, these natural constructions lead to bifiltered
simplicial complexes which are not $k$-critical for any finite $k$.
Nevertheless, we demonstrate that these bifiltrations can be computed
exactly and efficiently, and we provide a prototypical implementation
using the CGAL library. 
We also argue that 
some of the recent algorithmic advances for $2$-parameter persistence
(which usually assume $k$-criticality for some finite $k$)
carry over to the $\infty$-critical case.
\end{abstract}

\section{Introduction}
In the past years, the theory of \emph{multi-parameter persistent homology}
has gained increasing popularity. This theory extends the theory 
of (single-parameter) persistent homology by filtering a data set
with several scale parameters and observing how the topological properties
change when altering the ensemble of parameters. 
Most standard examples
define two scale parameters, where the first one is based on the distance
within the data set and the second one on its local density.
The motivation for that choice is an increased robustness against outliers
in the data set.

\subparagraph{Localized bifiltrations.}
We suggest a different type of bifiltration where the second parameter controls
the \emph{locality} of the data. Let $P$ be a finite set of data points in Euclidean space
$\R^d$ and $q\in\R^d$ a further point that we call the \emph{center}.
For two real values $s,r\geq 0$, we define
\[
L_{s,r}=\left(\bigcup_{p\in P} B_{s}(p)\right)\cap B_{r}(q)
\]
where $B_\alpha(x)$ is the ball of radius $\sqrt{\alpha}$ centered at $x$ (taking the square root
is not standard, but will be convenient later). In other words,
we consider the union of balls around the data points (as in many applications
of persistent homology), but we limit attention to a neighborhood around the center.
It is immediate that $L_{s,r}\subseteq L_{s',r'}$ for $s\leq s'$ and $r\leq r'$
so $L:=(L_{s,r})_{s,r\geq 0}$ is a nested sequence of spaces. 
We define the collection of spaces $L$ to be the \emph{absolute localized bifiltration}, see Figure~\ref{fig:initial_example} for an illustration.
The goal of this paper is to compute a combinatorial representation
of this bifiltration: a bifiltration of simplicial complexes
which is homotopy equivalent to the absolute localized bifiltration
at every choice $(s,r)$ of parameters.

Alternatively, we consider the variant where all points of $L_{s,r}$
on the boundary of $B_r(q)$ are identified (see Figure \ref{fig:absolute_vs_relative}). This version gives rise to a bifiltration that we call
the \emph{relative localized bifiltration}.
This sometimes reveals more local information around $q$
(as in Figure \ref{fig:absolute_vs_relative})
and is more frequently used in applications (see related work below).
Again, we are asking for an equivalent simplicial description.

\begin{figure}[ht]
\centering
\includegraphics[width=12cm]{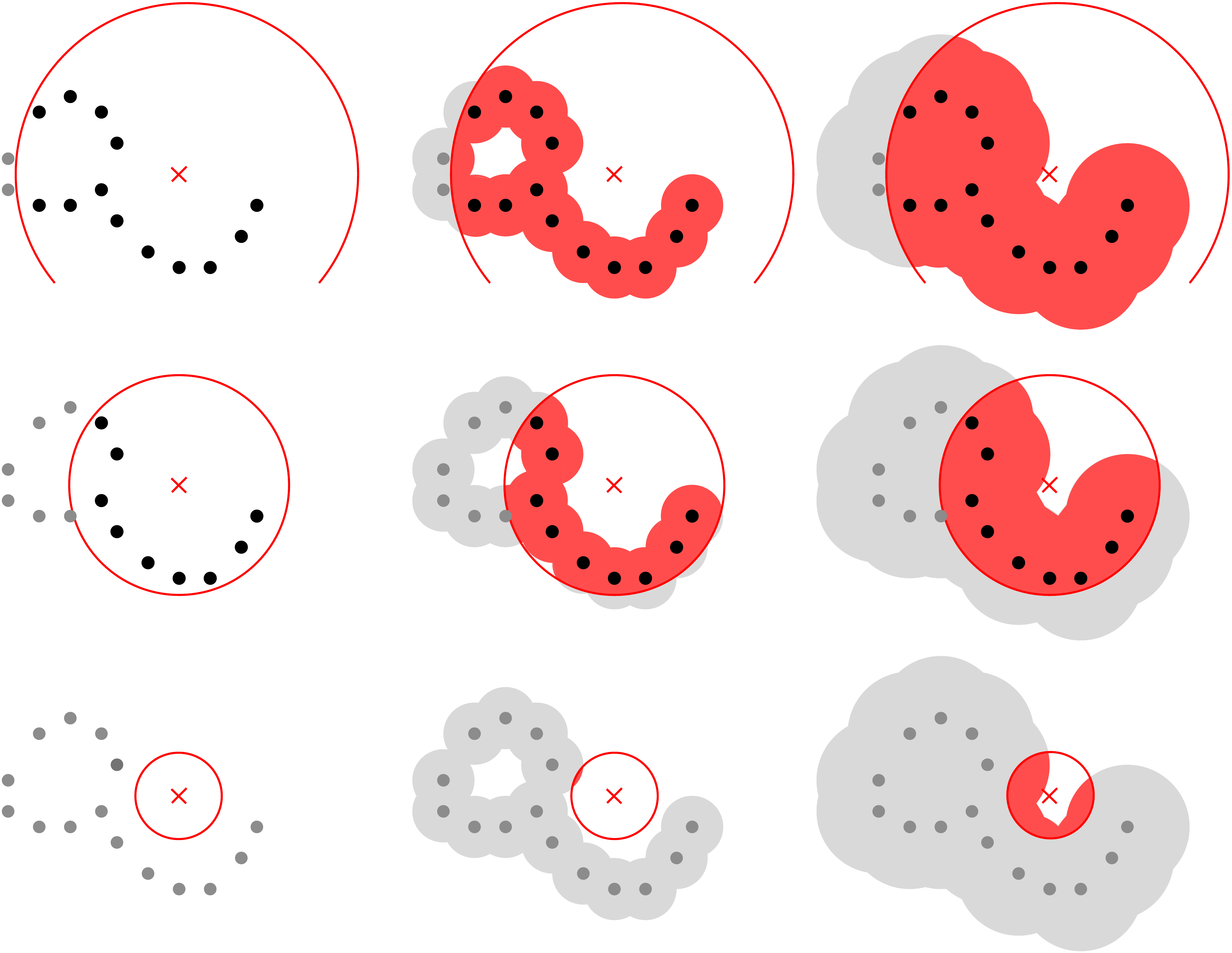}
\caption{Illustration of $L_{s,r}$. The radius $s$ controls the radius around the points in $P$ (black dots)
and grows in the horizontal direction. The radius $r$ of the center point (red cross) grows in vertical direction.
The sets $L_{s,r}$ are marked in dark, red color.}
\label{fig:initial_example}
\end{figure}

\begin{figure}[ht]
\centering
\includegraphics[width=8cm]{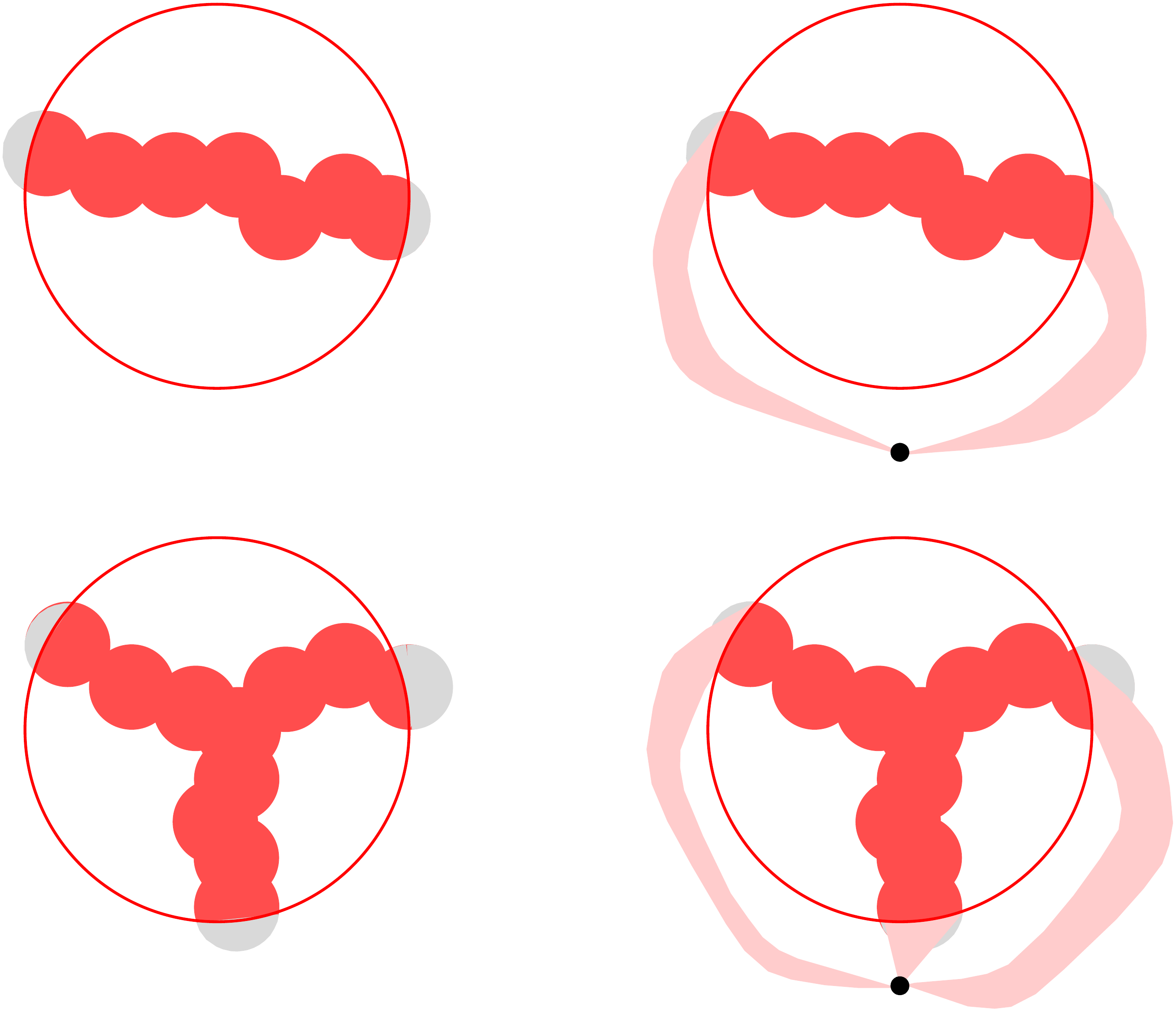}
\caption{Left: The two examples of $L_{s,r}$ are both contractible and therefore
cannot be distinguished further by homological methods. Right: The quotient space of $L_{s,r}$ relative to $L_{s,r}\cap \partial B_r(q)$ 
can be visualized by coning all points on the boundary of the ball with a (virtual) vertex,
drawn as a black dot here.
In the upper example, the resulting space has one hole, whereas the lower example has two. Therefore the homology of the spaces changes and allows for distinction.}
\label{fig:absolute_vs_relative}
\end{figure}

An interesting feature of these localized bifiltrations is that 
topological changes arise along curves in the two-dimensional parameter space spanned by $s$ and $r$. 
The perhaps simplest example is obtained by setting $d=1$, $q=0$ and $P=\{1\}$. Then, $L_{s,r}\neq\emptyset$
if and only if $r+s\geq 1$. Hence, the empty and non-empty regions
are separated by a line in the parameter space.
This implies that any equivalent simplicial bifiltration is 
\emph{$\infty$-critical}, meaning that there
is no integer $k$ for which it is $k$-critical. See Figure \ref{fig:criticality} for an illustration of $k$-criticality.

\begin{figure}[ht]
	\centering
	\includegraphics[width=3.5cm]{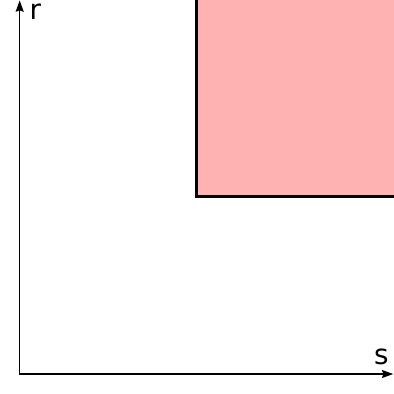}
	\hspace{1cm}
	\includegraphics[width=3.5 cm]{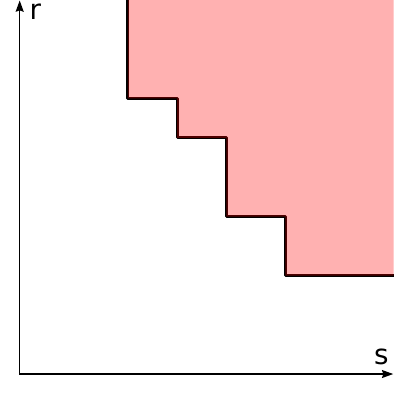}
	\hspace{1cm}
	\includegraphics[width=3.5 cm]{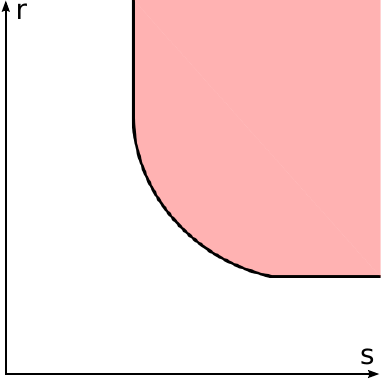}
	\caption{The active region and entry curve of a fixed simplex in the parameter space of a simplicial bifiltration, if it is $1$-critical (left), $4$-critical (middle) and $\infty$-critical (right). For a more formal definition of $k$-criticality, see, for instance,~\cite{cfklw-kernel}.}
	\label{fig:criticality}
\end{figure}

Given that $\infty$-critical simplicial bifiltrations are obtained from such a simple construction, we pose the question whether such bifiltrations allow for an efficient algorithmic treatment.
In this paper, we focus on the first step, the \emph{generation} of such bifiltrations.
This requires to compute, for every simplex its \emph{entry curve}, that is,
the boundary between the region of the parameter space where the simplex
is present and where it is not present (see Figure~\ref{fig:criticality}).
A natural idea might be to reduce to the $k$-critical case, approximating the entry curve by
a staircase with $k$ steps (a sequence of horizontal and vertical segments).
However, to ensure an accurate approximation,
the value of $k$ might be quite high which complicates the algorithmic treatment and introduces
another parameter to the problem.
Also, resorting to an approximation is unsatisfying, especially for the generation step which
is only the first step in the computational pipeline: while a discretization might suffice
for many tasks, it restricts the possibilities of subsequent steps. We therefore advocate
the computation of an \emph{exact} representation instead.

\subparagraph{Contributions.}
We give algorithms to compute absolute and relative localized bifiltrations exactly
and efficiently. The entry time of every simplex into the bifiltration is described by a curve in the 2-dimensional
parameter space that consists of line segments and parabolic arcs.

In the absolute case, a simplicial bifiltration is obtained 
using alpha complexes (also known as Delaunay complexes) and the Persistent Nerve Theorem. 
To determine the entry curve
of a (Delaunay) simplex, we solve a convex minimization problem on the dual Voronoi polytope
parameterized in $s$. We show that the solutions yield a polygonal chain
within the Voronoi polytope, and every line segment translates to one arc
of the entry curve in the parameter space. We also describe an efficient
algorithm to compute the entry curves of all simplices
in amortized constant time per simplex.

In the relative case, we use a variant of the Persistent Nerve Theorem for pairs.
However, the Voronoi partition does not satisfy the prerequisites of this Nerve theorem;
we show how
to subdivide the Voronoi cells for planar inputs to overcome this problem.
The entry curves of simplices are determined
by the same convex optimization problem as in the absolute case, but now asking
for a maximal solution, and can be treated with similar methods.

We provide a prototypical implementation\footnote{\url{https://bitbucket.org/mkerber/demo_absolute_2d/src/master/}} 
of the absolute case in the plane,
based on the \textsc{Cgal} library. This software allows us to visualize
the entry curves of localized bifiltrations and serves as a starting point
for subsequent algorithmic studies of $\infty$-critical bifiltrations.
We argue that an algorithmic treatment is in reach by showing that barcode
templates of such bifiltrations are computable with the same strategy 
as in the $1$-critical case.

\subparagraph{Related work.}
Applying the homology functor with field coefficients to absolute (relative) localized bifiltrations leads to persistence modules which we call absolute (relative) localized persistence modules. 
Horizontal and vertical slices of the relative localized persistence module are known as \emph{persistent local homology (PLH) modules} in literature, see the survey \cite{fw-plhsurvey}. 

A persistent version of local homology was first considered by Bendich et al. \cite{bceh-lhfromstratified} to infer the local homology of a stratified space given by a point cloud. Their PLH modules are defined via extended persistence diagrams \cite{cseh-extending} and the inherent two parameters ($s$ and $r$ in our notation) are taken into account through vineyards \cite{csem-vineyards}.
The study of stratified spaces with the help of PLH is continued in \cite{bwm-lhtransfer_strat}, where points of stratified spaces are clustered into same strata. In \cite{bceh-lhfromstratified} and \cite{bwm-lhtransfer_strat}, PLH modules are computed with modified alpha complexes. Skraba and Wang \cite{sw-approx} define two variants of PLH and show how both of them can be computed via approximations by Vietoris-Rips complexes. A further application is given by Ahmed, Fasy and Wenk who define a PLH based distance on graphs used for road network comparison \cite{afw-roadmapdistance}. 

The work mentioned above relies on relative versions of PLH. A persistence module similar to a horizontal slice of an absolute localized persistence module is used by Stolz \cite{s-landmark} for outlier robust landmark selection. Von Rohrscheidt and Rieck \cite{rr-toast} consider (samples of) tri-persistence modules to measure how far a given neighborhood of a point is from being Euclidean to obtain the "manifoldness" of point clouds. These tri-persistence modules are obtained by removing the open ball $B^o_t(q)$ from $L_{s,r}$. PLH modules of a filtration $(L_{s,r}\cap \partial B_r(q))_{s\geq 0}$ in combination with multi-scale local principal component analysis are used in \cite{bghirn-mlpca_plh} to extract relevant features for machine learning from a data set. Other applications of variants of persistent local (co)homology include \cite{dfw-dimdet}, \cite{wspvj-branching} and \cite{wbb-activationlandscapes}. 

Our work suggests bifiltrations for multi-parameter persistence. Other natural constructions of bifiltrations resulting from point-set inputs
are the density-Rips bifiltration~\cite{cz-multi}, the degree-Rips bifiltration~\cite{lw-rivet,rs-stable,rolle-degree}
and the multi-cover filtration~\cite{eo-multi,cklo-computing}; 
see~\cite{bl-stability} for a comparison of these approaches in terms of stability properties.
In all these approaches, the second scale parameter models the density of a point,
which is different from our approach where it rather models the locality with respect to a point $q$.

Computational questions in multi-parameter persistence have received a lot of attention recently,
including algorithms for visualization~\cite{lw-rivet}, decomposition~\cite{dx-generalized,boo-signed},
compression~\cite{lw-computing,fkr-compression,akp-filtration} and distances~\cite{klo-exact,kn-efficient,bk-asymptotic}.
In all aforementioned approaches, the input is assumed to be a simplicial bifiltration that is at least $k$-critical
(and usually even $1$-critical), so none of these algorithms is readily applicable to the filtrations
computed in this work.

The $\infty$-critical bifiltrations appearing in our work yield persistence modules that fulfill the tameness conditions of Miller\cite{m-homalgmodposets}. In \cite{m-homalgmodposets} the foundations to vastly generalize the aforementioned $k$-critical setting are laid. However, it is of rather theoretical nature and in contrast, we propose an inital setup for a concrete algorithmic treatment which might be followed along the suggested route in \cite{m-datastructures} (see Section $20$).

\section{The absolute case}
\label{sec:absolute}
We need the following concepts to formally state our problem:
For $s,r\in\R$, we write $(s,r)\leq (s',r')$ if $s\leq s'$ and $r\leq r'$.
A \emph{bifiltration} is a collection of topological spaces $X:=(X_{s,r})_{s,r\geq 0}$
such that $X_{s,r}\subseteq X_{s',r'}$ whenever $(s,r)\leq (s',r')$.
A bifiltration is \emph{finite simplicial} if each $X_{s,r}$ is a subcomplex
of some simplicial complex $\complex$.

A \emph{map of bifiltrations} $\phi:X\to Y$ 
is a collection of continuous maps $(\phi_{s,r}:X_{s,r}\to Y_{s,r})_{s,r\geq 0}$
that commute with the inclusion maps of $X$ and $Y$. Two bifiltrations are \emph{equivalent}
if there is a map $\phi:X\to Y$ such that each $\phi_{s,r}$ gives a homotopy equivalence (i.e., there exists a map $\psi:Y\to X$
with each $\psi_{s,r}$ being a homotopy inverse to $\phi_{s,r}$).

For a finite point set $P$ in $\R^d$, called \emph{sites} from now on, and a center point $q\in\R^d$ (not necessarily a site), 
we consider the bifiltration $L$ defined by
\[L_{s,r}=\left(\bigcup_{p\in P} B_s(p)\right)\cap B_r(q)\]
with $B_s(p)$ the set of points in distance at most $\sqrt{s}$ from $p$.
Our goal is to compute a finite simplicial bifiltration that is equivalent to $L$.

\subparagraph{Localized alpha complexes.}
The filtration $L_{s,\infty}$ is the union-of-balls filtration,
one of the standard filtration types in persistent homology (with one parameter).
It is also well-known that \emph{alpha complexes}~\cite{eh-computational}
provide a practically feasible way of computing an equivalent simplicial representation
(at least if $d$ is small). We summarize this technique next; the only difference
is that we localize the alpha complexes with respect to $B_r(q)$, which introduces
a second parameter but results in no theoretical problem:

For a site $p$, its \emph{Voronoi region} is the set of points
in $\R^d$ for which $p$ is a closest site:
\[\Vor(p):=\{x\in\R^d\mid \|x-p\|\leq \|x-p'\|\:\forall p'\in P\}.\]
Every $\Vor(p)$ is closed and convex. The \emph{restricted cover} $\cover_{s,r}:=\{U_{s,r}(p)\mid p\in P\}$ is given by
\[U_{s,r}(p):=B_s(p)\cap B_r(q)\cap\Vor(p).\]
For every $s,r\geq 0$, we have that $\bigcup_{p\in P} U_{s,r}(p)=L_{s,r}$. The \emph{localized alpha complex}
is the \emph{nerve} of $\cover_{s,r}$, that is, 
the abstract simplicial complex that encodes the intersection pattern of the restricted cover elements:
\[A_{s,r}:=\nrv\,\cover_{s,r}= \{\{p_0,\ldots,p_k\}\subseteq P\mid U_{s,r}(p_0)\cap\ldots\cap U_{s,r}(p_k)\neq\emptyset\}.\]
See Figure~\ref{fig:alpha_complex} for an illustration.
All $U_{s,r}(p)$ are closed and convex, and $U_{s,r}(p)\subseteq U_{s',r'}(p)$ for $(s,r)\leq (s',r')$.
With these conditions, the \emph{Persistent Nerve Theorem}~\cite{co-towards,bkrr-unified} ensures that
there is a homotopy equivalence between $L_{s,r}$ and $A_{s,r}$ for every $s,r\geq 0$, and moreover, these
homotopy equivalences commute with the inclusion maps of $L_{s,r}$ and $A_{s,r}$.
Hence, the bifiltrations are equivalent. 

\begin{figure}[ht]
\centering
\includegraphics[width=5cm]{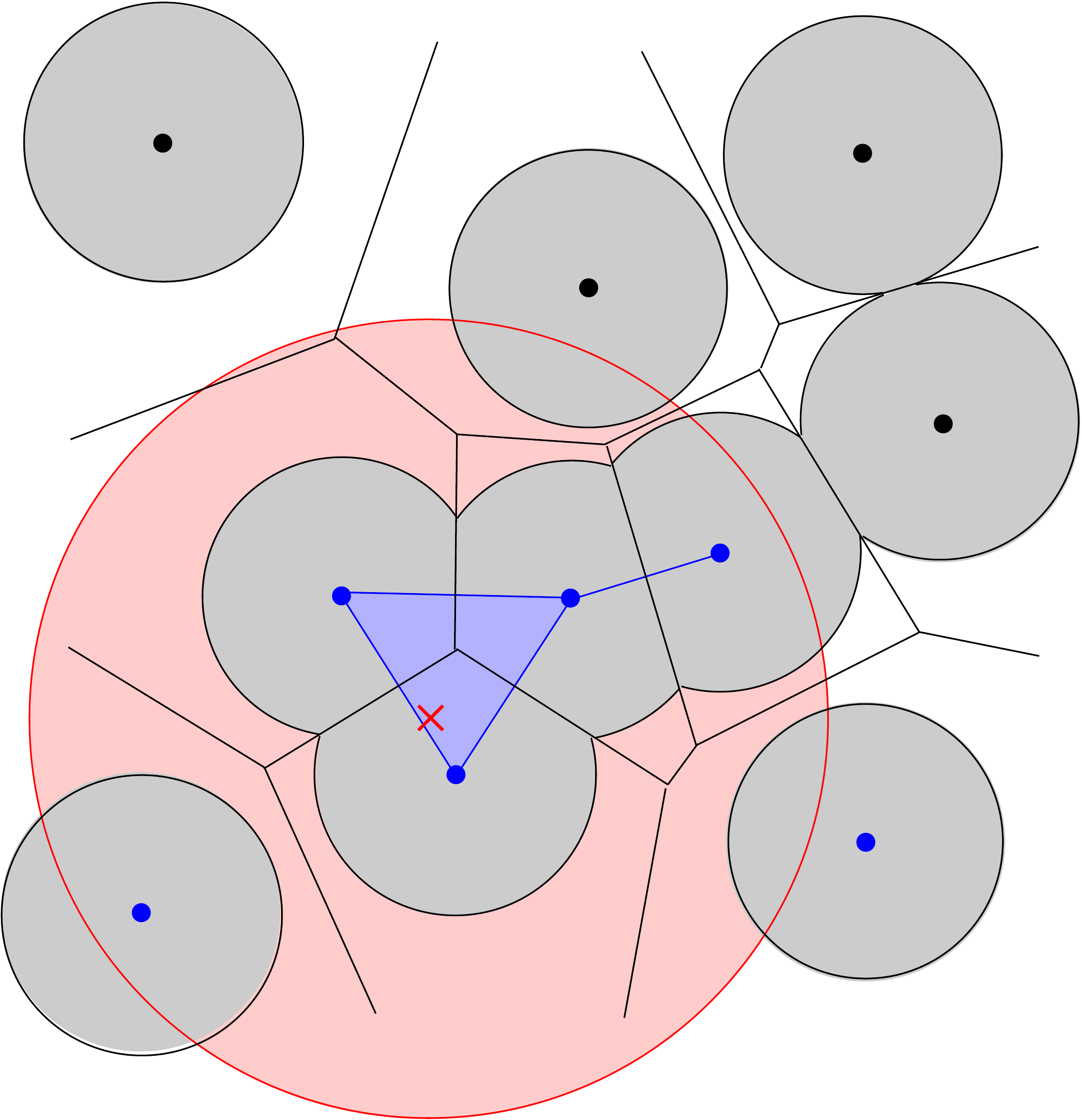}
\hspace{1cm}
\includegraphics[width=5cm]{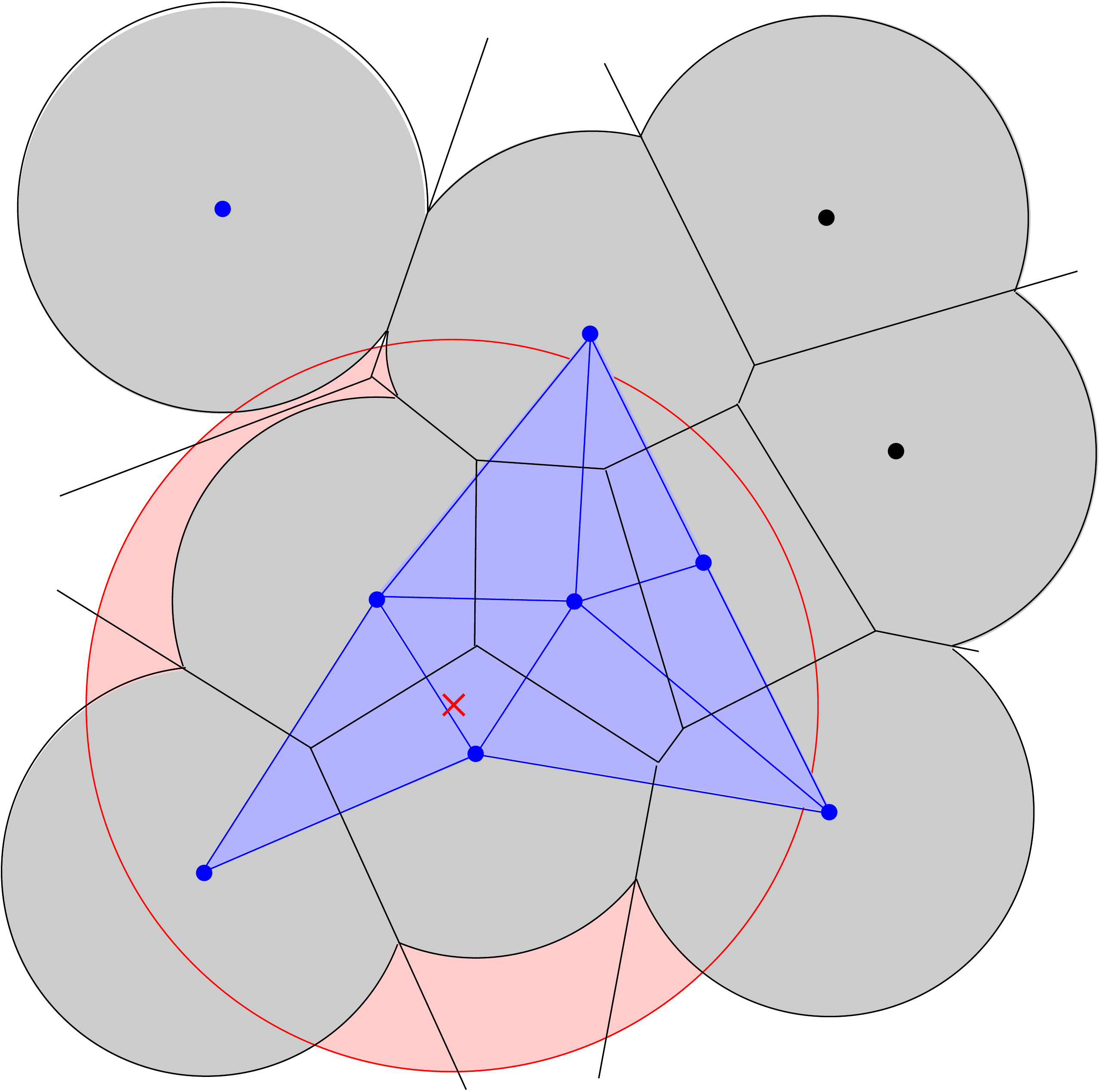}
\caption{The Voronoi regions, the restricted cover $U_{s,r}$ (in gray) and the localized alpha complex $A_{s,r}$
for two different choices of $(s,r)$. On the right, the upper blue vertices are indeed not connected although the
gray balls intersect, because they do not intersect within the (red) center circle.}
\label{fig:alpha_complex}
\end{figure}

Assuming that the sites are in generic position, that is, not more than $k+2$ sites lie on a common $k$-dimensional
sphere for $1\leq k\leq d-1$, the complexes $A_{s,r}$ are subcomplexes of the \emph{Delaunay triangulation} of $P$,
whose size is known to be at most $O(n^{\lceil d/2\rceil})$~\cite{akd-voronoi} and which can be computed efficiently, especially
in dimensions $2$ and $3$\cite{cgal:2dtriang,cgal:3dtriang}.

\subparagraph{Entry curves.}
We further specify what it means to ``compute'' the finite simplicial bifiltration $A$:
computing the Delaunay triangulation of $P$ yields all simplices belonging to $A$.
For each simplex $\sigma\in A$, we want to compute an explicit representation of
its \emph{active region}, defined as
\[R_\sigma:=\{(s,t)\in\R^2\mid \sigma\in A_{s,r}\}.\]
The active region is closed under $\leq$ in $\R^2$, meaning that
if $(x,y)\in R_\sigma$, the whole upper-right quadrant anchored at $(x,y)$ also belongs to $R_\sigma$.
We call the boundary of the active region the \emph{entry curve} of $\sigma$. See Figure \ref{fig:curves}
for an illustration of a family of entry curves.

\begin{figure}[ht]
	\centering
	\includegraphics[width=8cm]{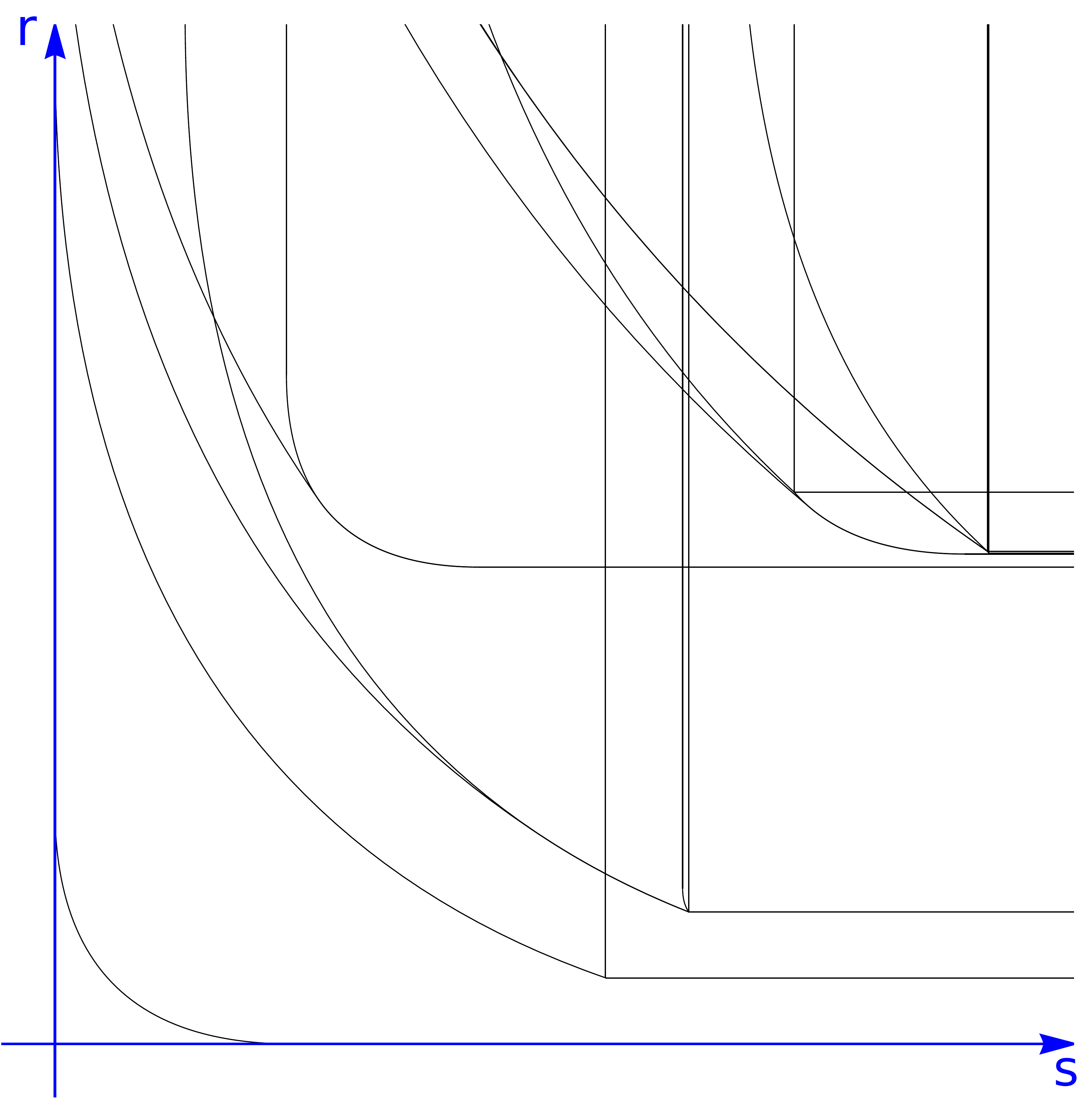}
	\caption{Section of entry curves of a planar absolute localized bifiltration on $25$ input points.}
	\label{fig:curves}
\end{figure}

To understand the structure of the entry curve, we define for $\sigma=(p_0,\ldots,p_k)$
the polytope $V_\sigma:=\Vor(p_0)\cap\ldots\cap\Vor(p_k)$. Then, letting p be some point of $\sigma$, we have that
$\sigma\in A_{s,r}$ if and only if $V_\sigma \cap B_s(p)\cap B_r(q)\neq\emptyset$.
Since $V_\sigma$ is convex, there is a unique point $\hat{p}\in V_\sigma$ as well as $\hat{q}\in V_\sigma$ with minimal distance to $p$ as well as $q$.
We write $s_0:=\|p-\hat{p}\|^2$, $s_1:=\|p-\hat{q}\|^2$
and $r_1:=\|q-\hat{q}\|^2$.
\begin{lemma}
The active region $R_\sigma$ lies in the half-plane $s\geq s_0$.
Moreover, restricted to the half-plane $s\geq s_1$, $R_\sigma$
is bounded by the line $r=r_1$, that is, the area above that line
is in $R_\sigma$ and the area below the line is not.
\end{lemma}
\begin{proof}
The first part follows because by definition, for $s<s_0$, the intersection $V_\sigma\cap B_s(p)$
is empty. Likewise, for $r<r_1$, the intersection $V_\sigma\cap B_r(q)$ is empty, so $R_\sigma$
is contained in the half-plane $r\geq r_1$.
Moreover, for $s\geq s_1$, the point $\hat{q}$ lies in $V_\sigma\cap B_s(p)$, so the intersection
$V_\sigma\cap B_s(p)\cap B_r(q)$ is non-empty if and only if $\hat{q}\in B_r(q)$, which
is equivalent to $r\geq r_1$.
\end{proof}
It remains to compute the entry curve in the $s$-range $[s_0,s_1]$.
For that, we want to compute for each such $s$, what is the minimal $r$-value for which
$V_\sigma \cap B_s(p)\cap B_r(q)\neq\emptyset$.
This minimal $r$-value, in turn, is simply the distance of $q$ to the set $V_\sigma\cap B_s(p)$.
We will study this geometric problem in the next subsection; the solution will give us a
parameterization of the boundary of the active region $R_\sigma$ by line segments and parabolic arcs.
\subparagraph{Minimizing paths.}
Slightly generalizing the setup of the previous paragraph, 
let $p,q\in\R^d$, and let $V$ be a closed convex polytope in $\R^d$, that is, the intersection
of finitely many closed half-spaces in $\R^d$. Let $\hat{p},\hat{q}$ be the points in $V$ with minimal distance to $p$
and $q$, respectively. Note that $p=\hat{p}$ is possible if $p\in V$, and $\hat{p}\in\partial V$ otherwise. This holds likewise for $q$. We set $s_0:=\|p-\hat{p}\|^2$ and $s_1:=\|p-\hat{q}\|^2$. For any $s\in [s_0,s_1]$,
the intersection $V\cap B_s(p)$ is not empty, and we let $\gamma_s$ denote the point in that intersection
that is closest to $q$. See Figure~\ref{fig:gamma_illustration} for an illustration of the case $d=2$.
The proofs of the next two statements are elementary and only exploit convexity of $V$ and that consequently,
the distance function to $q$ restricted to $V$ has only one local minimum.
The proofs are in Appendix~\ref{app:missing_proofs}.
\begin{figure}[ht]
\centering
\includegraphics[width=10cm]{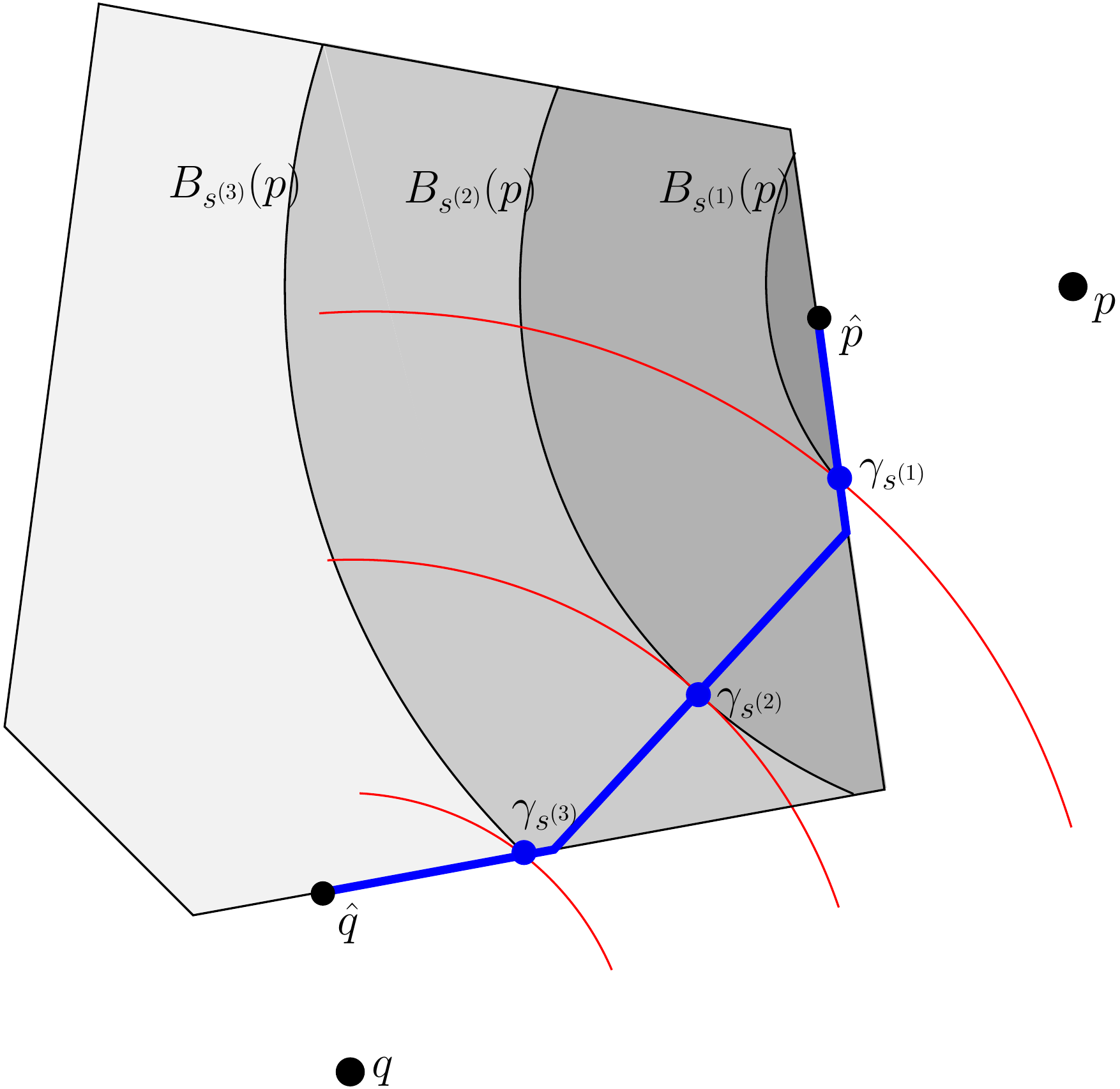}
\caption{For three radii $s^{(1)}<s^{(2)}<s^{(3)}$, the sets $V\cap B_{s_i}(p)$ and the corresponding points $\gamma_{s_1},\gamma_{s_2}, \gamma_{s_3}$ are illustrated. In fact, all
points $\gamma_s$ lie on the blue curve from $\hat{p}$ to $\hat{q}$.
The red arcs are arcs of a circle centered at q and indictate that $\gamma_{s_i}$ is indeed the minimizing point for $s_i$.}
\label{fig:gamma_illustration}
\end{figure}
\begin{lemma}
\label{lem:gamma_monotone}
We have that $\|p-\gamma_s\|^2=s$. In particular, the function $\|p-\gamma_s\|^2$ is
strictly increasing for $s\in [s_0,s_1]$. 
\end{lemma}
\begin{lemma}
\label{lem:continuous}
The function $\gamma: [s_0,s_1]\to \R^d, s\mapsto \gamma_s$ is continuous and injective.
\end{lemma}
It follows that $\gamma$ defines a path in $\R^d$ which we call the \emph{minimizing path} for $(V,p,q)$.
A minimizing path in the plane is displayed in Figure~\ref{fig:gamma_illustration} (in blue).
Because of the following lemma, we henceforth assume wlog that we only consider instances $(V,p,q)$ where $V$ is full-dimensional.
The statement follows easily by the Pythagorean Theorem~-- see Appendix~\ref{app:missing_proofs}.

\begin{lemma}[Dimension reduction]
\label{lem:dimension_reduction}
Let $V$ be a polytope in $\R^d$ contained in an affine subspace $W$. Let $p,q\in\R^d$ and $p',q'$
be the orthogonal projections of $p$ and $q$ to $W$, respectively. Then, the minimizing path of $(V,p,q)$
equals the minimizing path of $(V,p',q')$ up to a shift in the parameterization.
\end{lemma}
A \emph{face} of a convex polytope $V$ is the intersection of $V$ with a hyperplane $H$
such that all of $V$ lies in one of the closed half-spaces induced by $H$.
\begin{lemma}[Face lemma]
Let $\gamma$ be the minimizing path of $(V,p,q)$ and let $F$ be a face of $V$. 
Then, $\gamma\cap F$, the part of $\gamma$ that runs along $F$, is a subset of the minimizing path of $(F,p,q)$.
\end{lemma}
\begin{proof}
Let $x$ be a point on $\gamma\cap F$. By definition, $x=\gamma_s$ for some $s$, that is, $x$ is the closest point to $q$
in $V\cap B_s(p)$. Since $F\subseteq V$, $x$ is also the closest point to $q$ in $F\cap B_s(p)$, so $x$ lies on 
the minimizing path of $F$.
\end{proof}

By the Face Lemma, we know that the part of the minimizing path of $(V,p,q)$ that runs along $\partial V$ coincides
with the minimizing paths of its faces. It remains to understand the minimizing
path in the interior of $V$. The central concept to understand this sub-path is the following simple definition:
\begin{definition}
The \emph{bridge} of $(V,p,q)$ is the (possibly empty) line segment $V\cap\overline{pq}$,
where $\overline{pq}$ is the line segment of $p$ and $q$.
\end{definition}
\begin{lemma}[Bridge lemma]
Let $\gamma$ denote the minimizing path of $(V,p,q)$.
Then, the bridge is a subset of $\gamma$. Moreover, every point on $\gamma$ that does not belong to the bridge lies on $\partial V$.
\end{lemma}
\begin{proof}
Fix a point $x$ on the bridge and let $s:=\|x-p\|^2$. It is simple to verify that $s\in [s_0,s_1]$.
We argue that $x=\gamma_s$: indeed, the point $x$ is the closest point to $q$ in $B_s(p)$ (as it lies on $\overline{pq}$),
and since it also lies in $V$, it minimizes the distance to $q$ for the subset $V\cap B_s(p)$. That proves
the first part. 

For the second part, assume for a contradiction the existence of a point $y=\gamma_s$ for some $s\in [s_0,s_1]$
that is in the interior of $V$, but not on the bridge. Since $y$ must lie on the boundary of $V\cap B_s(p)$, and is not on $\partial V$,
it must lie in the interior of some spherical patch of $\partial B_s(p)$. The distance function to $q$, restricted to the $(d-1)$-dimensional 
sphere $\partial B_s(p)$
has no local minimum except at the intersection of $\overline{pq}$ with the boundary, but since $y$ does not lie on the bridge,
it is not that minimizing point. Hence, moving in some direction along the spherical patch
decreases the distance to $q$, contradicting the assumption that $y=\gamma_s$.
\end{proof}

The Bridge Lemma tells us that the minimizing path runs through the interior of $V$ at most along a single line segment, the bridge; before and after the bridge, it might have sub-paths on the boundary. Figure~\ref{fig:gamma_illustration} gives an example
where all three sub-paths are present.

\begin{theorem}[Structure Theorem]
\label{thm:structure_thm}
The minimizing path of $(V,p,q)$ is a simple path starting at $\hat{p}$ and ending at $\hat{q}$,
and every point on the path lies on the bridge of some face of $V$.
In particular, the path is a polygonal chain.
\end{theorem}
\begin{proof}
It is clear that the path goes from $\hat{p}$ to $\hat{q}$ and is simple because $\gamma$ is injective.
Since every point $x$ on the path lies in the relative interior of some face $F$, the Face Lemma implies that
$x$ lies on the minimizing path of $(F,p',q')$ (with $p'$, $q'$ the projections in the subspace of $F$),
and the Bridge lemma implies that $x$ lies on the bridge of $F$. Hence, the path is contained
in a union of finitely many line segments, and therefore is a polygonal path.
\end{proof}

\subparagraph{Algorithm.}
Let $\sigma$ be a Delaunay simplex, $V$ its Voronoi polytope, and $p$
one of the closest sites. We follow the natural approach to compute the
minimizing path for $(V,p,q)$ first. Then, for every line segment $\overline{ab}$
on the minimizing path, we use the parameterization
\begin{align}
\label{eqn:param1}
s=\|p-((1-t)a+tb)\|^2 \qquad \qquad r=\|q-((1-t)a+tb)\|^2,
\end{align}
which for $t\in [0,1]$ yields a branch of the entry curve.
Both $s$ and $r$ are quadratic
polynomials in $t$, therefore the resulting curve is a parabola or line -- see Appendix~\ref{app:algebra} for a simple proof.

To compute the minimizing path, we outline the construction
and defer details to Appendix~\ref{app:computation}:
we compute the points $\hat{p}$ and $\hat{q}$
which are the start- and endpoint of that path.
Then we compute the bridges of $V$ and of all its
faces. This yields a collection of line segments
and we compute the induced graph whose vertices are endpoints of bridges
(in this graph, bridges can be split into sub-segments if the endpoint
of another bridge lies in the interior). This graph can be directed
such that the distance to $p$ increases along every edge.
In this graph, we walk from $\hat{p}$ to $\hat{q}$ to compute the minimizing
path; the only required predicate is to determine the next edge to follow
at a vertex $x$. This is the edge
along which the distance to $q$ drops the most
which can be easily determined by
evaluating the gradient of the (squared) distance function to $q$
restricted to the outgoing edges.

We sketch the complexity analysis of this algorithm,
again defering to Appendix~\ref{app:computation} for details:
Let $N$ be the size of the Delaunay triangulation.
Computing all bridges over all Delaunay simplices is linear in $N$.
Writing $f$ for the number of faces of $V$, the graph constructed for $V$
consists of $f$ bridges. Since bridges do not (properly) cross,
the constructed graph for $V$ has still $O(f)$ edges, and the traversal
to find the minimizing path is done in $O(f)$ as well.
This immediately yields the complexity bound.
\begin{theorem}
\label{thm:complexity}
Let $P$ be $n$ points in general position in $\R^d$ where $d$ is constant.
Let $N$ be the size of the Delaunay triangulation of $P$.
We can compute the entry curves of all Delaunay simplices in time 
$O(N)$.
\end{theorem}

\section{Implementation}
\label{sec:implementation}
We have implemented the case of absolute localized bifiltrations in the plane
using the \textsc{Cgal} library. Our code computes the Delaunay triangulation~\cite{cgal:2dtriang}
and computes the minimizing path for each simplex, using a simplified
algorithm for the plane. 
For the minimizing path of a Delaunay vertex, we refer to Figure~\ref{fig:alg_2d} for an illustration:
The path is obtained by 
starting in $p$ and following the bridge until the boundary is hit. The path must continue in one of the two
directions along the boundary and follows the boundary until $\hat{q}$ is met (see the red path in Fig.~\ref{fig:alg_2d}). The decision on which direction
to follow can be answered by projecting $q$ to the supporting line of the boundary segment that is hit by the bridge 
and going towards that projection point $q'$ (the bridge might also hit the boundary in a vertex; we ignore
this degenerate case in our prototype).

All geometric predicates required for this implementation are readily available in the geometric kernel of \textsc{Cgal}~\cite{cgal:kernel}.
For the absolute, planar problem, it was possible to work exclusively with the Delaunay triangulation,
avoiding the explicit construction of the Voronoi polytope, but this will probably not be possible for other variants.
In either case, the exact geometric computation paradigm of \textsc{Cgal} guarantees that the obtained parameterization
is an exact representation of the simplicial bifiltration for the input data.

\begin{figure}[ht]
\centering
\includegraphics[width=8cm]{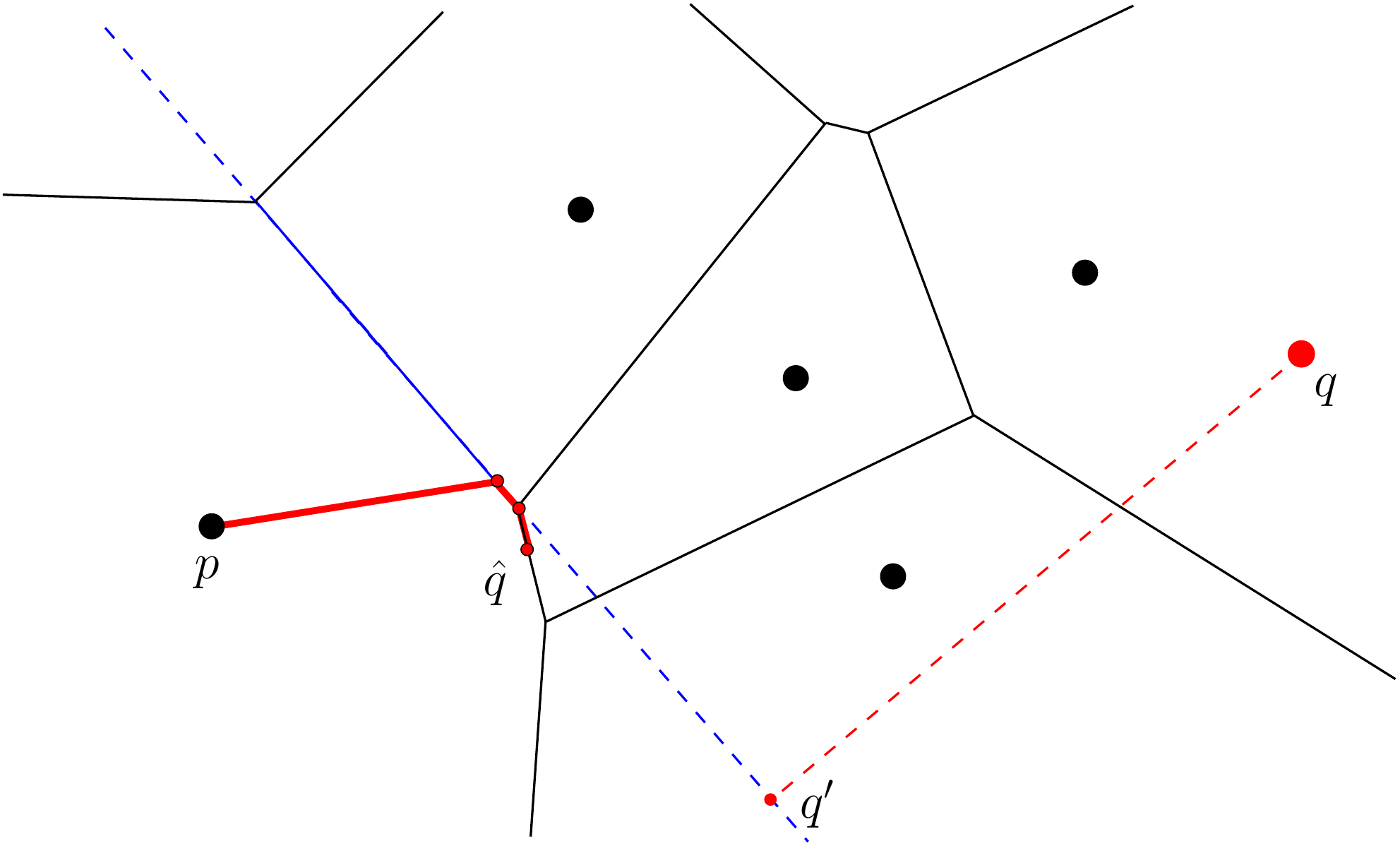}
\caption{Illustration of the minimizing path in the plane for a Voronoi region.}
\label{fig:alg_2d}
\end{figure}

\section{The relative case}
\label{sec:relative}
\subparagraph{Relative localized bifiltrations.} 
To better fit the situation studied in this section, we re-define $(s,r)\leq (s',r')$ if $s\leq s'$ and $r\geq r'$. 
With this poset, we can define bifiltrations, finite simplicial bifiltrations, and equivalence in an analogous way as in Section~\ref{sec:absolute}.
We extend these notions to pairs:
A \emph{bifiltration of pairs} is a collection of pairs of topological spaces (or simplicial complexes) $(X,A):=(X_{s,r},A_{s,r})_{s,r\geq 0}$ such that $X$ and $A$ are bifiltrations
and $A_{s,r}\subseteq X_{s,r}$. We have a \emph{finite simplicial bifiltration of pairs} if $X$ is a finite simplicial bifiltration,
and $A_{s,r}$ is a subcomplex of $X_{s,r}$ (which implies that also $A$ is finite simplicial).
We call two bifiltrations of pairs $(X,A)$ and $(Y,B)$ \emph{equivalent} if there is an equivalence $\phi:X\to Y$ such that 
the restriction maps $(\phi|_{A_{s,r}})_{s,r\geq 0}$ yield an equivalence of $A$ and $B$.

For a finite point set $P$, a center $q$ and $s,r\geq 0$, set $L_s=\bigcup_{p\in P} B_s(p)$ for the union of balls
and $\openball{r}{q}$ for the open ball around the center. Then the \emph{relative localized bifiltration} is the following collection of pairs of spaces
\[
\rb := (L_s,L_s\setminus\openball{r}{q})_{s,r\geq 0}.
\]
Indeed, it can be checked easily that this construction yields a bifiltration of spaces (notice the contravariance in the $r$-parameter).
Our computational task is to find a finite simplicial bifiltration of pairs equivalent to it.

What is the significance of this bifiltration? Note first that in general, a pair of bifiltrations $(X,A)$ induces
a (relative) persistence module $(H_n(X_{s,r},A_{s,r}))_{s,r\geq 0}$, leading to
\[
 H_n(\rb)=(H_n(L_s,L_s\setminus\openball{r}{q}))_{s,r\geq 0},
\]
which we call the \emph{localized relative persistence module}. 
By excision and the fact that all pairs are ``good pairs''~\cite{hatcher}, 
we have isomorphisms \cite{sw-approx}
\[
H_n(L_s,L_s\setminus\openball{r}{q})\cong H_n(L_s\cap B_r(q),L_s \cap \partial B_r(q))\cong \tilde{H}_n(L_s\cap B_r(q)/ (L_s \cap \partial B_r(q)))
\]
with $\tilde{H}$ denoting reduced homology. These isomorphisms imply that the relative localized persistence module corresponds point-wise to the space considered
in Figure~\ref{fig:absolute_vs_relative} (right). Moreover, the excision isomorphism commutes with the inclusion maps $L_s\subseteq L_{s'}$ for $s\leq s'$
which implies that the rows of the localized relative persistence module (with fixed $r$) are isomorphic to the module
$
(H_n(L_s\cap B_r(q),L_s \cap \partial B_r(q)))_{s\geq 0}
$
which was studied in \cite{bceh-lhfromstratified},\cite{bwm-lhtransfer_strat}, \cite{sw-approx}. Also the columns of the bifiltration (with $s$ fixed) have been studied earlier in \cite{sw-approx}.
Hence, the localized relative bifiltration encodes both types of modules of local persistent homology studied in previous work (see Appendix \ref{app:Homology} for a summary of basic notions).

We also remark that the computation of the relative localized persistence module can easily be reduced to the case of absolute homology
via the well-known coning construction~\cite[p.125]{hatcher}, yielding isomorphisms
$
H_n(X,A)\cong \tilde{H}_n(X\cup \omega\ast A)
$
for pairs of topological spaces $(X,A)$ with $\omega$ denoting a new vertex. These isomorphisms are functorial, yielding
an isomorphism between the relative persistence module of a pair of bifiltrations and the absolute persistence module of a bifiltration (using reduced homology).
Moreover, if the pair $(X,A)$ is finite simplicial, so is $X\cup \omega\ast A$.

\subparagraph{Nerves of pairs.}
To obtain an equivalent pair of finite simplicial complexes, we will define suitable covers for $L_s$ and $L_s\setminus\openball{r}{q}$
and construct the corresponding nerve complexes. For a (compact) topological space $X$, a cover $\cover_X$ of closed sets is \emph{good}
if the intersection of any subset of cover elements of $\cover_X$ is empty or contractible. In particular, a closed convex cover (as used in Section~\ref{sec:absolute})
is good because convex sets are contractible, and a non-empty intersection of convex sets is convex.

For a pair $(X,A)$ of spaces with $A\subseteq X$ closed, a closed cover $\cover_X$ of $X$ induces a cover $\cover_A$ of $A$ by restricting every cover element to $A$.
We say that the cover $\cover_X$ is \emph{good for the pair $(X,A)$} if $\cover_X$ and $\cover_A$ are both good covers. With this definition, we obtain the following
version of the Persistent Nerve Theorem; it follows directly from the results of \cite{bkrr-unified}
and we summarize the argument in Appendix~\ref{app:missing_proofs}.

\begin{theorem}[Functorial nerve theorem of pairs] 
	\label{thm:rel_nerve}
	Let $A\subseteq X \subseteq \mathbb{R}^d$ and $\cover_X$ be a good cover for the pair $(X,A)$. Then the spaces $X$ and $\nrv\,\cover_X$
are equivalent via a map $\phi:X\to\nrv\,\cover_X$,
and the restriction map of $\phi$ to $A$ yields a homotopy equivalence of $A$ and $\nrv\,\cover_A$.
Moreover, the map $\phi$ is functorial which means that if $X\subseteq X'$ and $A\subseteq A'$, the cover $\cover_{X'}$ is obtained by enlarging each
element of $\cover_X$ (or leaving it unchanged) and is a good cover for the pair $(X',A')$, then the maps $\phi$ and $\phi'$ commute with the inclusion maps $X\to X'$ and 
$\nrv\,\cover_X\to\nrv\,\cover_{X'}$, and the restricted maps to $A$ and $A'$ commute
with the inclusion maps $A\to A'$ and $\nrv\,\cover_A\to\nrv\,\cover_{A'}$.
\end{theorem}
This (rather bulky) theorem implies the following simple corollary in our situation: If we can find a filtration of (good) covers $\cover_{s}$ for the union of balls $L_s$,
such that for every $r\geq 0$, the induced cover $\cover_{s,r}$ on $L_s\setminus\openball{q}{r}$ is a good cover, then the pair $(\mathrm{Nrv}\,\cover_s,\mathrm{Nrv}\,\cover_{s,r})_{s,r\geq 0}$
is a finite simplicial bifiltration of pairs equivalent to the localized relative bifiltration.

\subparagraph{A good cover in the plane.}
A natural attempt to construct a good cover for $(L_s,L_s\setminus\openball{r}{q})$ would be to consider the cover induced by the Voronoi regions of $P$,
namely $\{\Vor(p)\cap B_s(p) | p \in P\}$,
similar as in Section~\ref{sec:absolute}.
While this cover is good for $L_s$, it is not good for the pair, because the removal of an (open) ball can lead to non-contractible cover elements and intersections;
Figure~\ref{fig:problems_with_cover} illustrates several problems.

\begin{figure}[ht]
\centering
\includegraphics[width=10cm]{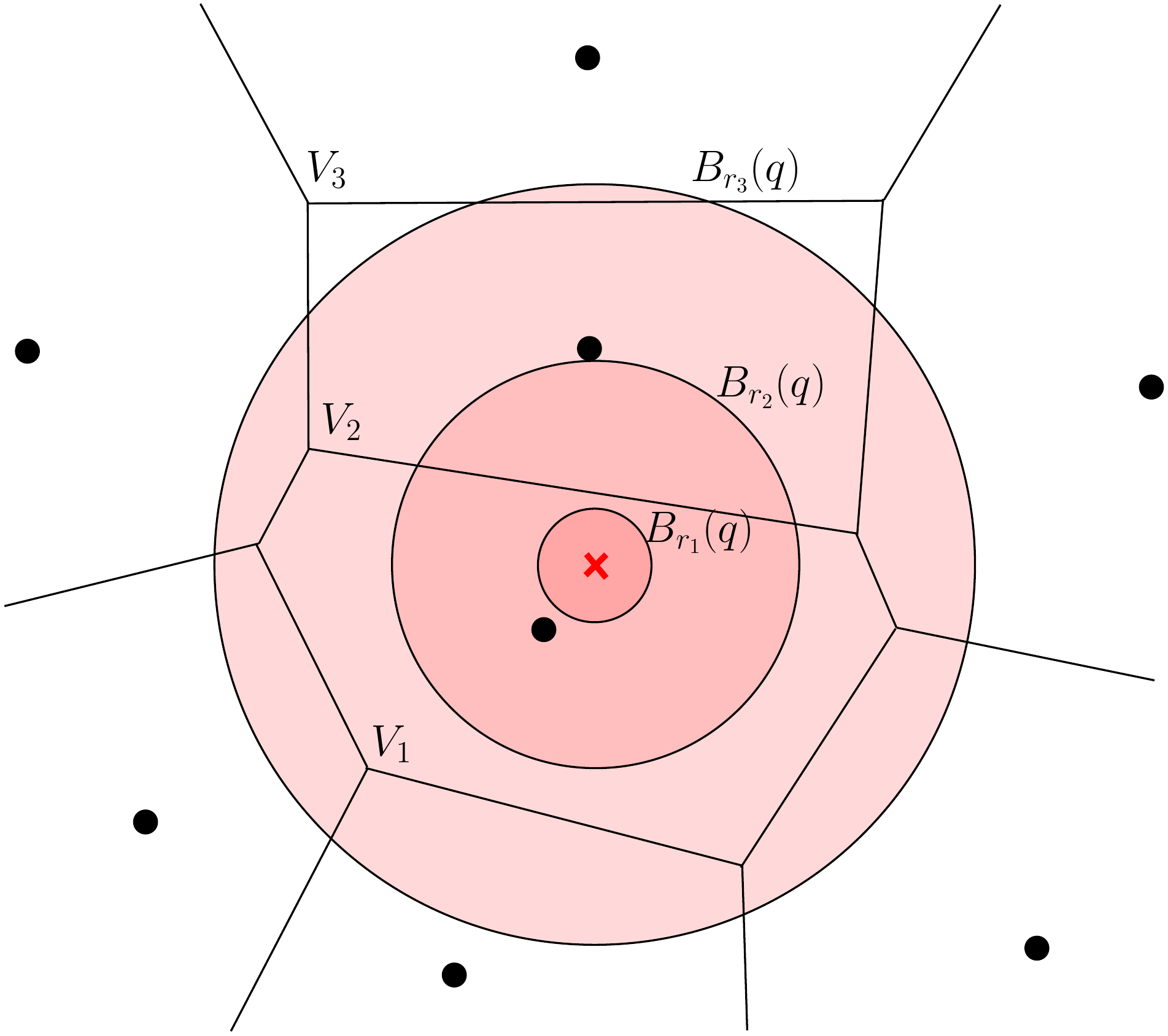}
\caption{Problems that arise when removing a ball $B_r(q)$ from the Voronoi cover: for radius $r_1$, the Voronoi region $V_1$ becomes non-simply connected.
For $r_2$, the Voronoi regions $V_1$ and $V_2$ intersect in two connected components. For $r_3$, the Voronoi region $V_2$ is disconnected.}
\label{fig:problems_with_cover}
\end{figure}

While the Voronoi cover itself does not work, we show that after suitably subdividing the Voronoi cells into convex pieces,
the induced cover is good for the pair $(L_s,L_s\setminus\openball{r}{q})$.
We will give the construction and will informally explain en passant why it removes all obstructions for being good
that are visible in Figure~\ref{fig:problems_with_cover}.

The construction has two parts: first, we split the Voronoi cell $\Vor(p)$ that contains $q$ into two pieces,
by cutting along the line $\overline{pq}$ (we assume for simplicity that $q$ lies in the interior of some Voronoi cell).
This initial cut ensures in particular that $q$ does not lie in the interior of a cover set anymore; this is 
a necessary condition because otherwise, removing a sufficiently small open ball $\openball{r}{q}$
yields a non-simply connected set in the induced cover for $L_s\setminus\openball{r}{q}$ (compare
radius $r_1$ in Figure~\ref{fig:problems_with_cover}). This cut also avoids
configurations where $\Vor(p)\cap B_s(p)$ gets disconnected because $B_r(q)$ is touching $B_s(p)$
from the inside; we refer to the proof of Lemma~\ref{lem:boundary} in Appendix~\ref{app:missing_proofs} for details.

To understand the second part of the construction,
we call a line segment or ray $e$ \emph{problematic} if
the distance function to $q$, restricted to e,
has a local minimum in the interior of $e$.
This is equivalent to the property that the orthogonal
projection of $q$ on the supporting line of the segment lies on $e$.
Observe in Figure~\ref{fig:problems_with_cover} that both for radius $r_2$
and $r_3$, the issue comes from a problematic edge (for $r_2$, the edge between $V_1$ and $V_2$
is problematic, for $r_3$, the edge between $V_2$ and $V_3$ is problematic).

In the second part of the construction, we go over all problematic edges in the Voronoi diagram of $P$.
For every such edge $e$, let $\hat{q}$ denote the orthogonal projection of $q$ on that edge. 
The line segment $\overline{q\hat{q}}$ cuts through one of the polygons $P$ incident to $e$,
and we cut $P$ using this line segment into two parts. This ends the description of the subdivision.
We denote the set of $2$-dimensional polygons obtained as $S_2(P,q)$.
See Figure~\ref{fig:subdivision} for an illustration. Note that the cuts introduced here
cut every problematic edge into two non-problematic sub-edges
and also avoids that a polygon gets disconnected when removing $B_r(q)$.

\begin{figure}[ht]
\centering
\includegraphics[width=8cm]{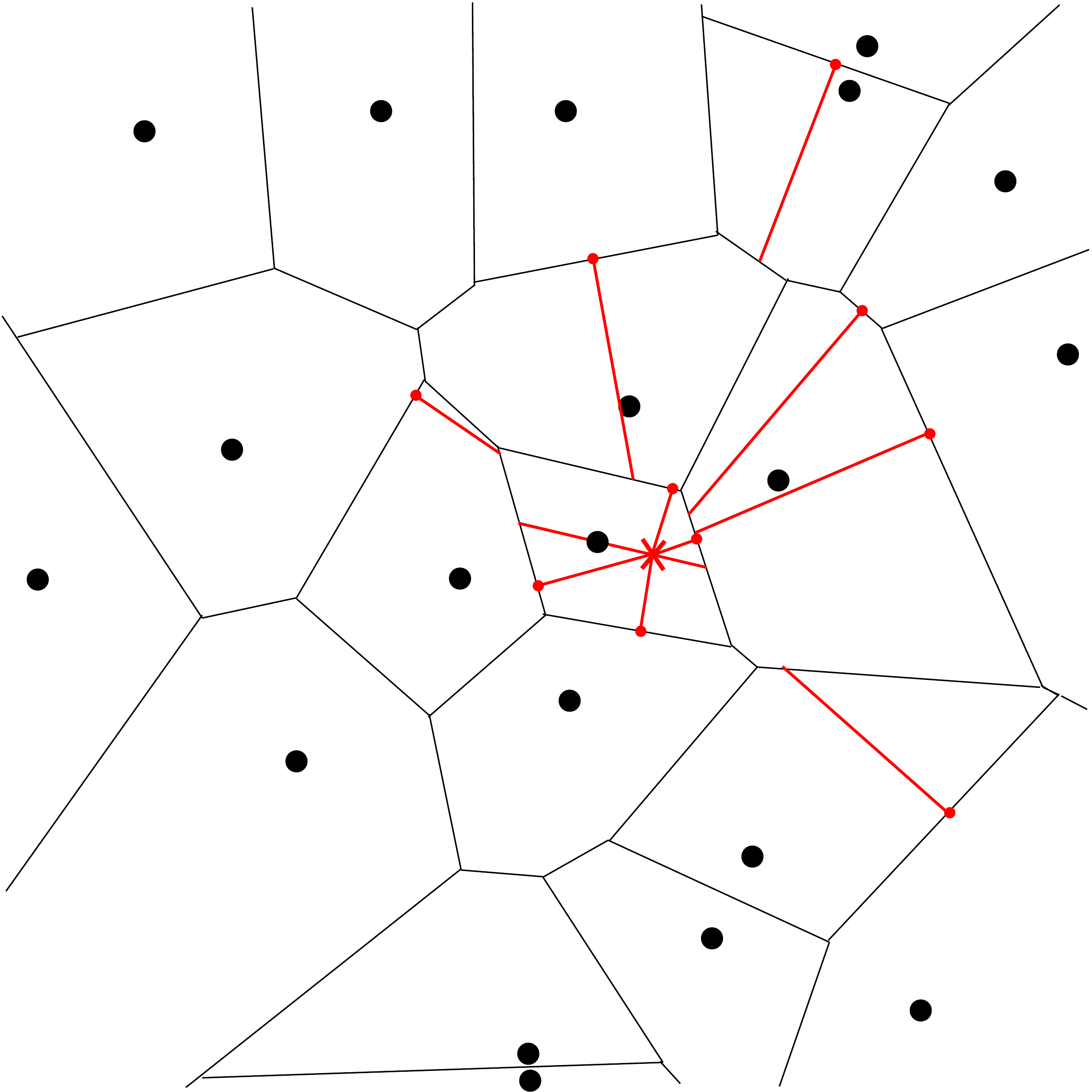}
\caption{Illustration of the subdivision $S_2(P,q)$. The subdivision lines
are in red.}
\label{fig:subdivision}
\end{figure}

Importantly, all cuts introduced are line segments on lines through $q$. This implies that no two of the cuts
can cross.  Moreover, no edge arising from a cut can be problematic. This leads to the following theorem,
whose complete proof is given in Appendix~\ref{app:missing_proofs}
\begin{theorem}
\label{thm:good_cover}
The cover $\{L_s\cap V\mid V\in S_2(P,q)\}$ is good for $(L_s,L_s\setminus\openball{r}{q})_{s,r\geq 0}$.
\end{theorem}

Every cut increases the number of polygons by one, and the number
of edges by at most $3$. Since we introduce at most one cut per Voronoi edge, plus one initial cut,
the complexity of the subdivision is $O(n)$. 

\subparagraph{Entry curves.}
It is left to compute the entry curves of every simplex in the nerve induced by $S_2(P,q)$.
The described method generalizes also to higher dimensions assuming a subdivision inducing a good cover
is available, but we keep the description planar for simplicity.
In this case, every simplex $\sigma$ represents a polygon, edge, or vertex of the planar subdivision defined by $S_2(P,q)$,
which we denote by $V$.
Note that $\sigma$ has two entry curves: one for $L_s$ and one for $L_s\setminus\openball{r}{q}$.
The former is simple to describe: since $L_s$ does not depend on $r$, the entry curve is determined by the line
of the form $s=\|p-\hat{p}\|^2$, with $p$ a closest site to $V$ and $\hat{p}$ the closest point to $p$ in $V$.

For $L_s\setminus\openball{r}{q}$, we fix $s$ and search for the largest $r$ such that $B_s(p)\cap V\setminus\openball{r}{q}$
is not empty. This is equivalent to finding the point on $V\cap B_s(p)$ with \emph{maximal} distance to $q$,
which is the opposite problem to what was considered in Section~\ref{sec:absolute}.
For brevity, we restrict the discussion to the case of bounded polytopes $V$, postponing
the (straight-forward) extension to unbounded polytopes to Appendix~\ref{app:computation}.
We set $q^\ast$ as a point in $V$ with maximal distance to $q$, $s_0:=\|p-\hat{p}\|^2$ 
and $s_1:=\|p-q^\ast\|^2$. For $s\in [s_0,s_1]$, we define $\Gamma_s$ to be a point with maximal
distance to $q$ in the set $V\cap B_s(p)$.

\begin{figure}[ht]
\centering
\includegraphics[width=6cm]{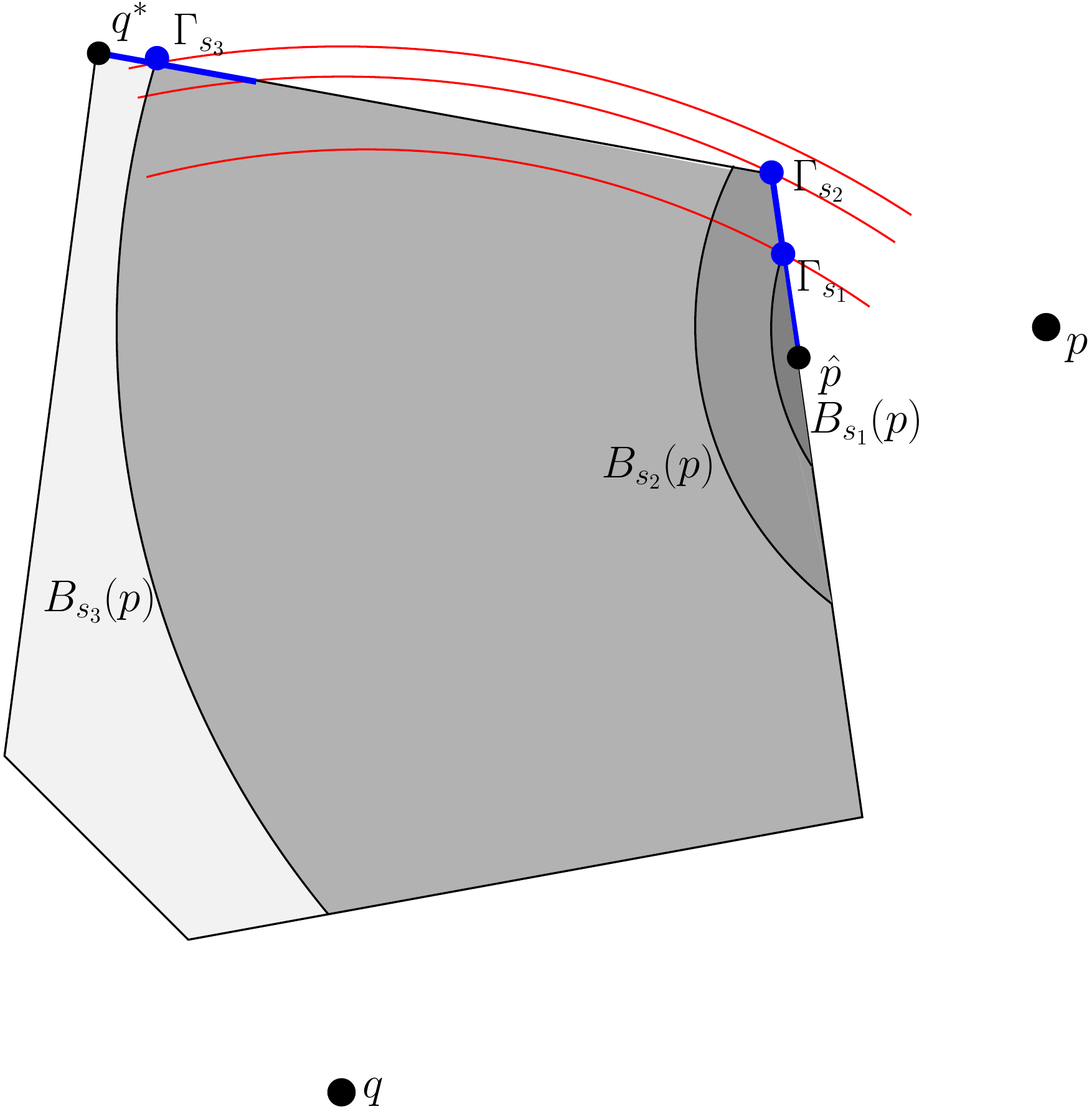}
\caption{Illustration of $\Gamma_s$ for the same $(V,p,q)$ as in Figure~\ref{fig:gamma_illustration}. 
The point $\Gamma_{s_2}$ remains the maximal point for a range of $s$-values,
and the maximizing ``path'' (blue) jumps at some value of~$s$.}
\label{fig:gamma_max_illustration}
\end{figure}

The maximizing problem is not as well-behaved as the minimization problem:
for general polytopes $V$, the point $\Gamma_s$ might not be unique, and there can be discontinuities in the image of $\Gamma$; 
see Figure~\ref{fig:gamma_max_illustration} for an example.
However, such problems are caused by local extrema on the boundary
of $B_s(p)\cap V$ which can be excluded for the polytopes of $S_2(P,q)$. Hence,
we can define the \emph{maximizing curve} $\Gamma:[s_0,s_1]\to \Gamma_s$ and infer (proof in Appendix):
\begin{lemma}
\label{lem:Maximal_continuous}
For $V\in S_2(P,q)$, the point $\Gamma_s$ is unique and the curve $\Gamma$ is continuous.
\end{lemma}
The structure of the curve $\Gamma$ is similar to its minimizing counterpart $\gamma$. The Dimension Reduction
Lemma and the Face Lemma also hold for $\Gamma$, with identical proofs. We define the \emph{anti-bridge}
as the intersection of $V$ with the ray emanating from $p$ in the direction $-pq$ (see Figure~\ref{fig:anti_bridge}).
With that, we obtain the statement that the anti-bridge is part of $\Gamma$, and the rest of $\Gamma$ lies
on $\partial V$, with an analogue proof as for the bridge lemma. This leads to

\begin{figure}[ht]
\centering
\includegraphics[width=4cm]{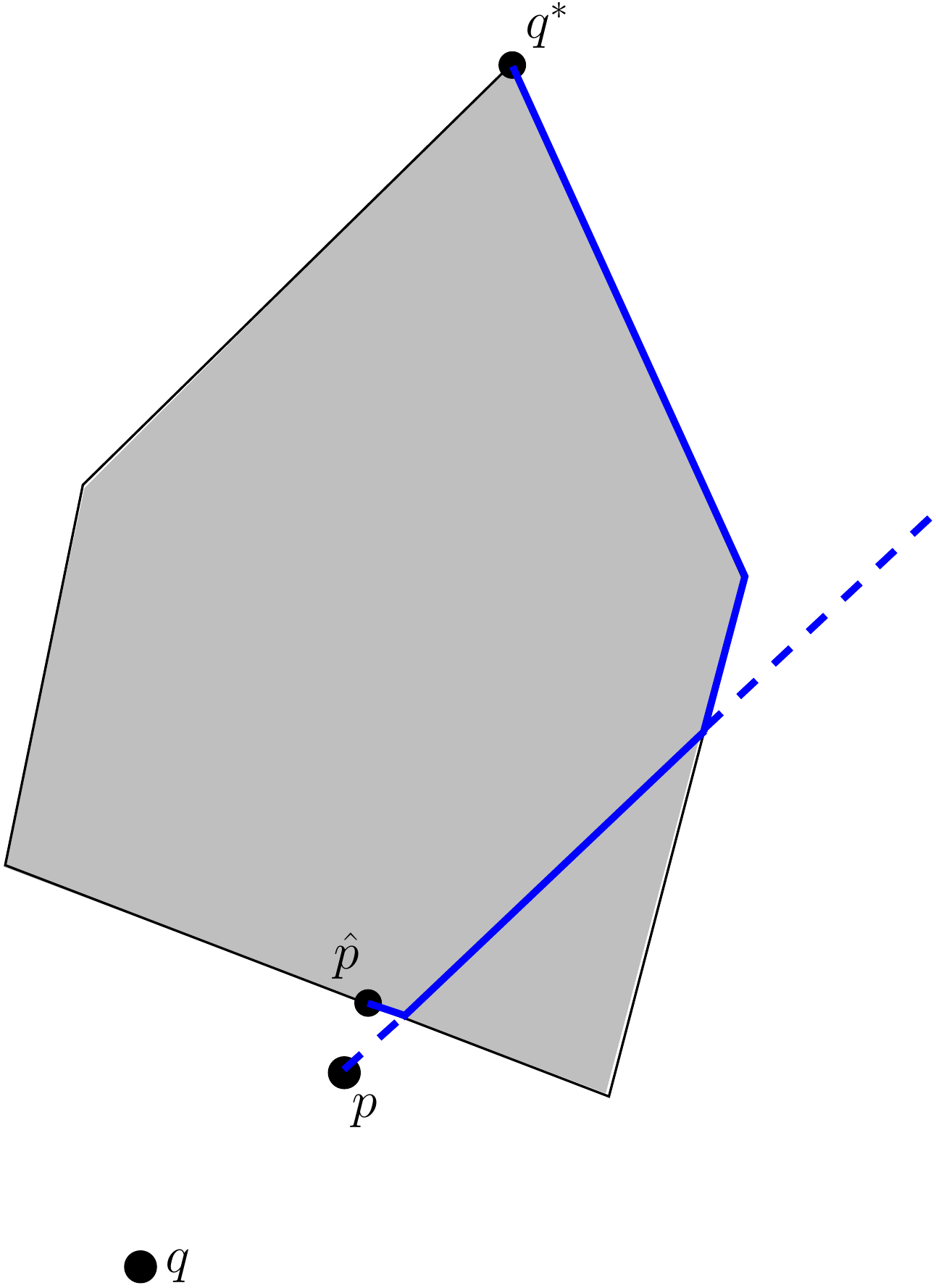}
\caption{Illustration of the maximizing path $\Gamma$ for a polygon with only one local minimum and one local maximum
on the boundary. The path goes through the interior along the anti-bridge.}
\label{fig:anti_bridge}
\end{figure}

\begin{theorem}[Structure Theorem, maximal version]
\label{thm:structure_max}
The maximizing path $\Gamma$ is a subset of the union of all anti-bridges over all faces of $V$. In particular,
it is a polygonal chain.
\end{theorem}
The theorem also infers a way to compute the entry curve of $\sigma$
constructing and traversing a graph obtained by anti-bridges. We omit further details which are analogous to the minimization case, and yield the same bound as in Theorem~\ref{thm:complexity}.

\section{\texorpdfstring{Barcode templates for $\infty$-critical bifiltrations}{Barcode templates for infinity-critical bifiltrations}}
We demonstrate in this section that representing a bifiltration via non-linear entry curves
does not prevent an efficient algorithmic treatment. We focus on the case of the
computation of \emph{barcode templates} as introduced by Lesnick and Wright~\cite{lw-rivet}.
The idea is as follows: restricting a bifiltration to a line of positive slope (called a \emph{slice})
gives a filtration in one parameter, and hence a persistence barcode. While we cannot associate a persistence barcode
to the bifiltration~\cite{cz-multi}, the collection of barcodes over all slices yields
a wealth of information about the bifiltration. The software library RIVET\footnote{\url{https://rivet.readthedocs.io/en/latest/index.html\#}}
provides a visualization of these sliced barcodes for finite simplicial bifiltrations. To speed up the visualization step,
RIVET precomputes all \emph{combinatorial barcodes}, that is, it clusters slices together on which the simplices enter the filtration
in the same order~-- it is well-known that the barcode combinatorially only depends on this order, not the concrete critical values.
The barcode template is then, roughly speaking, the collection of all these combinatorial barcodes.
Barcode templates have also been used for exact computations of the matching distance of bifiltrations~\cite{klo-exact,bk-asymptotic}.

All aforementioned approaches assume the bifiltration to be $1$-critical, which means in our notation that the active region
of every simplex $\sigma$
is the upper-right quadrant of a point $v_\sigma\in\R^2$.
We argue that this restriction is unnecessary: indeed, the combinatorial barcode is determined by the order in which the slice
intersects the entry curves of all simplices. For a finite set $M$ of points in $\mathbb{R}^2$, a non-vertical line partitions $M$ into $(M_\uparrow,M_\downarrow,M_{on})$,
denoting the points above the line, below the line, and on the line, respectively.

\begin{lemma}
Let $I$ denote the set of intersection points of pairs of entry curves. If two slices have the same partition of $I$, they have the same
combinatorial barcode.
\end{lemma}
\begin{proof}
Let $\ell_1$, $\ell_2$ be slices with different barcodes. This means, that there is at least one pair of simplices $(\sigma,\tau)$
for which the entry curves are intersected in a different order. When continuously transforming $\ell_1$ into $\ell_2$, we therefore
have to cross an intersection point of these entry curves, and hence the partition changes.
\end{proof}

Using standard point-line duality~\cite[Ch.~8]{dutch-book} which is known to preserve above/below orders of points and lines,
we obtain at once:

\begin{corollary}
Let $\mathrm{dual}(I)$ be the line arrangement obtained by the dual of all points in $I$. Every region of this arrangement
is dual to a set of slices with the same combinatorial barcode.
\end{corollary}

We mention that this result generalizes the well-studied $1$-critical case: 
the entry curve of every $\sigma$ consists of a vertical and a horizontal ray emanating
from a point $v_\sigma$ in this case, and the curves of $\sigma$ and $\tau$ intersect in the join of $v_\sigma$ and $v_\tau$.
Hence $I$ is the union of all pairwise joins of the critical values of all vertices.

Computationally, the only complication consists in computing the intersection points of the entry curves, which have a more
complicated structure than in the $1$-critical case. However, if the curves are semi-algebraic
(and of small degree), as in the case studied in this paper, computing these intersection
points is feasible and efficient software is available~\cite{bhk-generic}.

\section{Conclusion}
Localized persistence modules are, on the one hand, a suitable set of examples 
for further algorithmic work in the pipeline of $\infty$-critical multi-parameter
persistence. On the other hand, because one-dimensional sub-modules
of the module have been considered in the context of applications,
we hope that the module itself can be of use in applied contexts.

There are many natural follow-up questions for our work: most directly,
the only obstacle to extend our relative approach in higher dimensions
is the subdivision scheme which is currently only proved for $\R^2$.
Another future task is the implementation of our approaches beyond
the absolute case in the plane. Such implementations are well in-reach,
given that all underlying geometric primitives are provided
by the \textsc{Cgal} library.

One aspect we have not touched upon is stability. 
A natural assumption is that if $P$ and $P'$ as well as the centers $q$ and $q'$ 
are $\eps$-close in the $L_\infty$-distance, then the localized bifiltrations
for $(P,q)$ and $(P',q')$ are $\eps$-interleaved~\cite{ccsggo-proximityofpers,l-theoryofinterleaving}. 
This claim is true (and straight-forward to prove) if we re-define the notion $B_\alpha(x)$
to be ball of radius $\alpha$ around $x$ (instead of $\sqrt{\alpha}$
as in this work). While our non-standard notion for balls prevents us from 
making this claim, we point out that it simplifies the description of the
entry curves. In the standard notion, the parameterization 
in (\ref{eqn:param1}) in Section~\ref{sec:absolute} 
would involve square roots and therefore does not yield parabolic arcs. We decided
to favor the simplicity of the entry curves in this work.

We hope that our initial efforts will result in an algorithmic
treatment of multi-parameter persistence that can cope with $\infty$-critical
filtrations with comparable efficiency as in the $1$-parameter case.
A natural next question is whether minimal presentations
of $\infty$-critical filtrations can be computed efficiently,
generalizing the recent approaches from~\cite{lw-computing,fkr-compression}.

Localizing the union-of-balls filtration in a center yields more fine-grained
information about the point set in the vicinity of the center. We pose the
question whether and how this information can be leveraged in the numerous
application domains of topological data analysis. We speculate that even in
situations where the data set does not contain a canonical center location,
considering a sample of centers and analyzing the ensemble of localized
bifiltrations yields a more discriminative topological proxy then
the union-of-balls filtration.

\bibliography{bib}

\begin{appendix}

\section{Missing proofs}
\label{app:missing_proofs}

\subparagraph{Proof of Lemma~\ref{lem:gamma_monotone}.}
Note first that the statement is true for $s=s_1$ by definition.
Also by definition, $\gamma_s\in B_s(p)$, so $\|p-\gamma_s\|^2\leq s$. Assume for a contradiction that $\|p-\gamma_s\|^2<s$
for some $s\in [s_0,s_1)$. This implies that $\gamma_s$ lies in the interior of $B_s(p)$ so
there exists an $\eps>0$ small enough, such that $B_\eps(\gamma_s)\subset B_s(p)$, and therefore also
$B_\eps(\gamma_s)\cap V \subset B_s(p)\cap V$. However, the distance function to $q$,
restricted to $V$ has no local minimum except at $\hat{q}$ because of convexity,
and $\gamma_s\neq\hat{q}$ for $s<s_1$. Therefore, there exists a point in $B_\eps(\gamma_s)\cap V$ that is closer to $q$
than $\gamma_s$, contradicting the definition of $\gamma_s$.

\subparagraph{Proof of Lemma~\ref{lem:continuous}.}
For continuity, consider a sequence $(s_n)_{n\in\N}$ that converges to $s$.
The sequence $(\gamma_{s_n})$ is bounded because every $\gamma_s$ lies in $B_{s_1}(p)$.
It hence suffices to show that any accumulation point of the sequence $(\gamma_{s_n})$
equals $\gamma_s$.

So, consider an accumulation point $x\in\R^d$ and assume for a contradiction that $x\neq\gamma_s$. 
Let $(t_n)$ denote a subsequence
of $(s_n)$ such that $(\gamma_{t_n})$ converges to $x$. Since all $\gamma_{t_n}$
lie in $V$, and $V$ is closed, it follows that $x\in V$. Moreover, since $\|p-\gamma_{t_n}\|^2=t_n$
by the previous lemma and because the squared distance function to $p$ is continuous, we
have that $\|p-x\|^2=s$. Hence, $x\in V\cap B_s(p)$. Since $\gamma_s\in V\cap B_s(p)$ as well,
by convexity, the line segment $x\gamma_s$ lies in that set too. Notice that the squared distance function
to $p$ restricted to segments has no local maximum, thus the squared distance to $p$ is strictly smaller than
$s$ throughout the interior of $x\gamma_s$.

Since the closest point to $q$ in $V\cap B_s(p)$ is unique and $x\neq\gamma_s$ is assumed, we have
that $\|q-x\|^2>\|q-\gamma_s\|^2$. The continuity of the distance function implies that there exist
neighborhoods $U_x$ around $x$ and $U_{\gamma_s}$ around $\gamma_s$ such that the distance from $q$ to
any point in $U_x$ is strictly greater than the distance from $q$ to any point in $U_{\gamma_s}$.
Pick a point $y\in U_{\gamma_s}\cap x\gamma_s$. As said before, its squared distance to $p$ is $s'<s$.
On the other hand, since the $\gamma_{t_n}$ converge to $x$ and $t_n$ converge to $s$, we can find
some $m\in\N$ such that $s'<t_m$ and $\gamma_{t_m}$ lies in $U_x$. By definition, $\gamma_{t_m}$
is the closest point to $q$ in $V\cap B_{t_m}(p)$. But $y$ lies in $V\cap B_{s'}(p)\subseteq V\cap B_{t_m}(p)$
as well, and its distance to $q$ is smaller, a contradiction. This proves the the continuity.
Injectivity follows directly from the previous lemma.

\subparagraph{Proof of Lemma~\ref{lem:dimension_reduction}.}
By the Pythagorean Theorem, the point minimizing the distance to $q$ also minimizes the distance to $q'$,
implying that the minimizing paths of $(V,p,q)$ and $(V,p,q')$ are the same. Furthermore, with $\delta:=\sqrt{\|p-p'\|}$,
we have that $V\cap B_s(p)=V\cap B_{s-\delta}(p')$. Hence, the minimizing paths of $(V,p,q')$ and $(V,p',q')$
differ only in a parameter shift by $\delta$.

\subparagraph{Proof of Theorem \ref{thm:rel_nerve}.}
We use the notation of~\cite{bkrr-unified}. According to their Definition 1.2, in the category
of covered spaces, there is a morphism $(A,U_A)\to (X,U_X)$. This gives rise to the following
commutative diagram of spaces

\[ \begin{tikzcd}
  A \arrow[hookrightarrow]{d} & \arrow[swap]{l}{\rho^{(A)}_S} \mathrm{Blowup}(U_A) \arrow{r}{\rho^{(A)}_N}\arrow[hookrightarrow]{d} & \mathrm{Nrv} (U_A)\arrow[hookrightarrow]{d} \\
X & \arrow[swap]{l}{\rho^{(X)}_S} \mathrm{Blowup}(U_X) \arrow{r}{\rho^{(X)}_N} & \mathrm{Nrv} (U_X) \\%
\end{tikzcd}
\]
where $\mathrm{Blowup}(X)$ is the blowup complex of the cover, which can be thought of as a product space of the space $X$ and the nerve $\mathrm{Nrv}(U_X)$,
and $\rho_S$, $\rho_N$ are the canonical projection maps. 
Because the vertical arrows are inclusions, it follows that $\rho^{(A)}_S$ is the restriction of $\rho^{(X)}_S$ to $A$, and the same holds
for the $\rho_N$ maps.
Since the covers $U_A$ and $U_X$ are both good, the maps $\rho_N$ are homotopy equivalences
(Thm 5.9, part 2a in~\cite{bkrr-unified}). Moreover, our covers are ``sufficiently nice'' to guarantee all prerequisites of Corollary 5.16 in~\cite{bkrr-unified},
certifying that also the maps $\rho_S$ are homotopy equivalences. Finally, taking the homotopy inverse $\tau^{(X)}_S$ to $\rho^{(X)}_S$, we can observe
at once that its restriction to $A$ is a homotopy inverse to $\rho^{(A)}_S$. The same holds for the homotopy inverse of $\rho^{(X)}_N$, proving the first part
of the theorem. The functoriality follows from exactly the same diagram and argumentation, using the morphisms between the covered spaces $(X,U_X)\to (X',U_{X'})$
and $(A,U_A)\to (A',U_{A'})$, respectively.

\subparagraph{Proof of Theorem~\ref{thm:good_cover}.}
The major technical result required is that for any polygon in $S_2(P,q)$, no
additional local extrema appear when intersecting with $L_s$.

\begin{lemma}
\label{lem:boundary}
For each polygon $V\in S_2(P,q)$, if $V\cap L_s\neq\emptyset$, the distance function to $q$ restricted to the boundary of $V\cap L_s$ 
contains only one local minimum. Consequently, it contains also only one local maximum, and the set
\[\partial(V\cap L_s)\setminus\openball{r}{q}\]
is empty or connected for every $r\geq 0$.
\end{lemma}

With that statement, it is simple to prove that the cover induced by $S_2(P,q)$ is good for the pair $(L_s,L_s\setminus\openball{r}{q})$:
First of all, it is clearly good for $L_s$ since all induced cover elements are convex.
For $L_s\setminus\openball{r}{q}$, the intersection of two cover elements is a subset of the intersection of two polygons in $S_2(P,q)$.
By our construction, no such edge has a local minimum of the distance function in its interior, implying that removing $\openball{r}{q}$
leaves a connected line segment as intersection, implying that the intersection of two cover elements is contractible.
Higher-order intersections happen only in points which is good too.

We are left to argue that each cover element is contractible. So, fix $V\in S_2(P,q)$ and consider, for $(s,r)$ fixed,
the induced cover element $V\cap L_s\setminus\openball{r}{q}$.  
Every point in this set is path-connected to a point of $\partial(V\cap L_s)$ that is not in $\openball{r}{q}$,
just by moving in the direction opposite to $q$ until the boundary is hit. Lemma~\ref{lem:boundary} certifies
that all these boundary points are connected to each other, implying that the cover element is connected.

It remains to show that every cover element is simply connected. Indeed, $V\cap L_s$ is simply connected,
and the only possibility to change this would be if $\openball{r}{q}$ would entirely lie with $V$.
However, the initial cut in the construction of $S_2(P,q)$ makes sure that $q$ is not in
the interior of any polygon. This completes the proof of Theorem~\ref{thm:good_cover};
it is only left to show Lemma~\ref{lem:boundary}.

\begin{proof}[Proof of Lemma ~\ref{lem:boundary}]
We first observe that $V\cap L_s=V\cap B_s(p)$ for a (unique) site $p$
because $V$ is subset of a Voronoi region of some site $p$. Hence, $\partial(V\cap B_s(p))$ consists
of line segments which are part of $\partial V$ and circular arcs that belong to $B_s(p)$. Clearly,
the value $\gamma_s$, which is defined as the (unique) point in minimal distance to $q$, lies on the boundary
(otherwise, $q$ would be inside $V$ which is excluded by the construction of $S_2(P,q)$) and is a local minimum.
We argue that there is no other local minimum.

Assume that $x$ is a local minimum on $\partial(V\cap B_s(p))$.
Let $r:=\|q-x\|^2$. There are three cases for the location of $x$ which all lead to the conclusion that $x=\gamma_s$:
\begin{itemize}
\item $x$ lies in the interior of a line segment $e$ of $\partial V$. This implies that $x$ is the projection of $q$ on the Voronoi edge $e'$
that contains $e$, and this edge is problematic by definition. By construction, a cut is performed at $x$, for the polygon adjacent to $V$
that is closer to $q$. Because $x$ lies in the interior of a line segment of $\partial V$ by assumption, $V$ must be the polygon
further away from $q$. But that implies that $x$ is the only intersection point of $B_r(q)$ and $V$, and so, $x=\gamma_s$.
\item $x$ lies at a vertex of $\partial (V\cap B_s(p))$. There are two sub-cases: if $x$ is a point where two line segments meet (with an angle
of less than $\pi$), then $B_r(q)$ can only touch $V$ in $x$ if $x$ is the only point in the intersection of $V$ and $B_r(q)$.
This means again that $x=\gamma_s$. If $x$ is a point where a line segment and a circular arc of $B_s(p)$ meet, the angle
between the line segment and the tangent of the circle in $x$ is necessarily smaller than $\pi$, and the same argument applies.
\item The last case is that $x$ lies in the interior of a circular arc of $B_s(p)$.
Note that the line $pq$ intersects $B_s(p)$ in two points, and the
only possible location of a local minimum is the intersection point closer to $q$. Now, if $q$ lies inside $B_s(p)$, this
implies that $q$ is closer to $p$ than $x$, and since $x$ lies in the Voronoi region of $p$, so does $q$. It follows
that $V$ arose from the Voronoi region that contains $q$, and the initial cut of $S_2(P,q)$ has cut that region using the
line $pq$. This is a contradiction to the assumption that $x$ lies in the interior of a circular arc.
Hence, $q$ lies on or outside $B_s(p)$, and thus $x$ lies on the line segment $\overline{pq}$. Therefore, $x$ lies on the
bridge of $(V,p,q)$, and the Bridge lemma implies that $x=\gamma_s$.
\end{itemize}

This proves the first part of the Lemma. Since the $\partial(V\cap B_s(p))$ is homeomorphic to a circle, two local maxima
would imply two local minima, contradicting the first part. Finally, since $\partial(V\cap L_s)\setminus\openball{r}{q}$
is simply a superlevel set of the distance function, two components would imply the existence of two local maxima,
one in each component.

\end{proof}

\subparagraph{Proof of Lemma~\ref{lem:Maximal_continuous}.}
The following two properties of $V\in S_2(P,q)$ are direct consequences of Lemma~\ref{lem:boundary}.

\begin{lemma}
The furthest point to $q$ on $V\cap B_s(p)$ is unique.
\end{lemma}
\begin{proof}
Note first that every furthest point must necessarily be on the boundary, and is also a furthest point on the boundary.
If there were two distinct furthest points, they are in particular local maxima of the distance function to $q$
on the boundary, which contradicts Lemma~\ref{lem:boundary}.
\end{proof}

\begin{lemma}
The distance to $q$ has no local maximum except at $q^\ast$.
\end{lemma}
\begin{proof}
Assume that there is a second local maximum. Again, both that maximum and $q^\ast$ are also local maxima on the boundary.
For $s$ large enough, this implies that the two points are also local maxima of $V\cap B_s(p)$, again
contradicting Lemma~\ref{lem:boundary}.
\end{proof}

Re-inspecting the proof of Lemma~\ref{lem:gamma_monotone}, the second lemma lets us invoke
the analogous argument to prove that $\|p-\Gamma_s\|^2=s$ for all $s$ such that $\|p-\hat{p}\|^2\leq s\leq \|p-q^\ast\|^2$.
The first lemma guarantees that the proof of Lemma~\ref{lem:continuous} is valid also for the maximizing path $\Gamma$.

\section{Computational details}
\label{app:computation}
\subparagraph{Building bridges.}
We have given a point set $P\subset \R^d$, a Delaunay simplex $\sigma=\{p_0,\ldots,p_k\}$ and a $q\in\R^d$
with $d$ a constant. We write $V$ for the intersection of the Voronoi regions of $p_0,\ldots,p_k$. We set $p:=p_0$ and
want to construct the bridge of $(V,p,q)$.

First, we compute the projections $p'$ and $q'$ of $p$ and $q$ to the subspace $S$ that supports $V$. This can be done by iteratively projecting
to bisector hyperplanes $(p,p_1),\ldots,(p,p_k)$ in $O(k)=O(d)=O(1)$ time.
$V$ is the intersection of half-spaces in $S$, and they can be easily obtained
through the co-facets of $\sigma$: If $\{p_0,\ldots,p_k,p_{k+1}\}$ is a co-facet,
the corresponding half-space is bounded by the intersection of the bisector of $p$ and $p_{k+1}$ and contains $p$.

Let $f=f_\sigma$ be the number of co-facets of $\sigma$. To compute the bridge,
we parameterize the line through $p'$ and $q'$ via $tp'+(1-t)q'$
and compute  for every bounding half-space of $V$ its intersection
with the $p'q'$-line (which requires to compute the intersection of the line
with the bounding hyperplane in $S$). This yields a parameter range
$(\infty,t_0]$ or $[t_0,\infty)$, and the intersection of all these
ranges yields an interval $[t_1,t_2]$ which we intersect with $[0,1]$
to obtain the bridge. This costs $O(f)$ time.

Hence, with $N$ the number of simplices in the Delaunay triangulation,
the total cost to compute all bridges is
\begin{align*}
&\sum_{\sigma\in\mathrm{Del}(P)} O(f_\sigma) = O(\sum_{\sigma\in\mathrm{Del}(P)} \#\{\text{co-facets of }\sigma\})=O(\sum_{\sigma\in\mathrm{Del}(P)} \#\{\text{facets of }\sigma\})\\&=O(\sum_{\sigma\in\mathrm{Del}(P)} (d+1))=O(N).
\end{align*}

\subparagraph{Details on graph construction.}
Again, let $V$ be a Voronoi polytope. 
We construct a graph $G$ whose vertices are the endpoints of bridges of faces of $V$
and whose edges are induced by the bridges. 
Note that a bridge runs through the interior of a face except at its endpoints.
Since the interiors of faces are disjoint, it follows that two
bridges cannot cross. 
However, the endpoint
of a bridge $b_1$ can lie in the interior of another bridge $b_2$, so 
a bridge might contribute several edges to $G$. In this case, we say that
the $b_1$ \emph{splits} $b_2$. Every split increases the number of
edges by one. We argue that a bridge can, however, only be split 
at most $2^{d+1}$ times, which is a constant as $d$ is constant.
Indeed, the bridge of a polytope
$W'$ can only split the bridge of $W$ if $W'$ is a co-face of $W$.
But the number of co-faces of $W$ is at most $2^{d+1}$
because $W$ is a Voronoi polytope given by $k\leq d+1$
sites $\{p_1,\ldots,p_k\}$, and each co-facet is given by a subset of these sites.

The algorithm to compute $G$ is then straight-forward: 
Let $\sigma$ be the Delaunay simplex that corresponds to $V$.
For every co-face $\tau$ of $\sigma$ in the Delaunay complex,
we split its bridge by iterating over all faces of $\tau$ that
are co-faces of $\sigma$ and check whether one of the bridge endpoints
lies on the bridge of $\tau$. If so, we split that bridge accordingly.
After having finished splitting the bridge of $\tau$, we add the resulting
subdivided bridge segments into $G$. By the analysis above, with
$m$ the number of faces of $V$, constructing $G$ has a complexity of $O(m)$.

\subparagraph{Details on computing the minimizing path.}
Given the graph $G$ as above for $V$, we compute the minimizing path for
$(V,p,q)$. We first consider $G$ as directed, such that along every edge,
the distance to $q$ is decreasing. By definition and the Structure Theorem (Theorem~\ref{thm:structure_thm}), the minimizing path must be
a directed path along $G$. 

We first compute the starting point $\hat{p}$ of the path.
This point is a vertex of $G$: by the Structure Theorem, it is contained 
in some bridge, and it cannot lie in the interior of that bridge because
then, moving along the bridge towards $p$ would yield
a point in $V$ with smaller distance to $p$, contradicting the minimality
of $\hat{p}$. Hence we can find $\hat{p}$ in $O(m)$ time, with
$m$ the number of faces of $V$, by just computing
the distance to $p$ for each vertex of $G$ and taking the minimum.
The same argument shows that also $\hat{q}$ can be computed in $O(m)$ time.

To construct the minimizing path, we simply walk along edges of $G$
starting from $\hat{p}$ until we encounter $\hat{q}$. At every vertex $x$
on this path, we have to decide 
which outgoing edge defines the minimizing path.
For each outgoing edge, the squared distance to $q$ decreases, 
and the minimizing
path follows the edge where the squared distance has the steepest decrease.
This can be easily determined by the derivative: For each outgoing edge $e$,
let $\ell$ be the line supporting the edge. The squared distance function
to $q$ along $\ell$ is a quadratic function, so its derivative is a line.
We evaluate the derivative at $x$ for all outgoing edges
and follow the edge with the smallest derivative. If two or more
derivatives are equal, we choose the edge with smallest second derivative
(which is a constant). Note that no two edges can have equal first
and second derivative, 
as this would imply that the squared distance to $q$ 
remains constant, contradicting the uniqueness of the minimizing path.

The complexity of finding the next edge is then proportional
to the number of outgoing edges. Since every edge is considered at most in
one step of the walk, the complexity of computing the minimizing path
is bounded by $O(m)$.
 
Putting everything together, we need $O(N)$ time (with $N$ the number 
of Delaunay simplices) to compute the bridges and $O(m_\sigma)$
time for computing the minimizing path for $\sigma$, with $m_{\sigma}$ 
the number of co-faces of $\sigma$. Hence the total complexity is

\begin{align*}
&O(N)+\sum_{\sigma\in\mathrm{Del}(P)} O(m_\sigma) = O(N+\sum_{\sigma\in\mathrm{Del}(P)} \#\{\text{co-faces of }\sigma\})\\ &=O(N+\sum_{\sigma\in\mathrm{Del}(P)} \#\{\text{faces of }\sigma\})=O(N+ N\cdot 2^{d+1})=O(N).
\end{align*}

\subparagraph{$\Gamma_s$ for unbounded polygons.}
If $V$ is unbounded, $q^\ast$, the furthest point from $q$ in $V$, does not exist. Instead, it is simple to see
that $\|q-\Gamma_s\|\to\infty$ for $s\to\infty$. Hence, unlike in the absolute case and unlike in the case of bounded polygons,
the entry curve does not become constant eventually, but contains a non-constant arc towards $\infty$ in $s$-direction.

Nevertheless, the curve $\Gamma$ remains a continuous injective curve also in this case, and the Structure
Theorem (Theorem~\ref{thm:structure_max})
holds also in this case with identical proof, except that the polygonal chain contains one ray.
This ray is either the anti-bridge, if an infinite part of it lies in $V$, or it is a part of an unbounded
boundary edge of $V$. The parameterization of the ray is then given as in 
(\ref{eqn:param1}) from Section~\ref{sec:absolute}, with $a$
the starting point of the ray, $b$ any further point of the ray, and $t$ ranging in $[0,\infty)$.
This still yields a ray or (unbounded) parabolic arc.

\section{Parameterized curves}
\label{app:algebra}
The goal of this section is to prove:

\begin{theorem}
Let $p,q,a,b$ be points in $\R^d$. Then the parameterized curve
\begin{align}
s&=\|p-((1-t)a+tb)\|^2\\
r&=\|q-((1-t)a+tb)\|^2
\end{align}
with $t\in [0,1]$ is the arc of a parabola or of a line.
\end{theorem}

\begin{proof}
Rewriting the parameterization in vector form, we obtain
\[
\left(\begin{array}{c}s\\r\end{array}\right)=u t^2 + v t + w
\]
with $u,v,w\in\R^2$. If $u$ and $v$ are lineary dependent, we obtain
the equation of a line, as one can easily check. Otherwise,
there is an affine transformation that maps $(u,v)$ to $(e_1,e_2)$,
and after a suitable translation, the parameterization is
\[
\left(\begin{array}{c}s\\r\end{array}\right)=e_1 t^2 + e_2 t=\left(\begin{array}{c}t^2\\t\end{array}\right)
\]
which is a parabola. The proof is completed by noting that affine transformations map parabolas to parabolas~\cite[Thm 12.12]{bls-euclidean}.
\end{proof}

\section{Homology}
\label{app:Homology}

\subparagraph{Relative homology.}
Given a topological space $X$ and a subspace $A\subseteq X$, the induced boundary maps on quotient chain groups $\partial_n: C_n(X)/C_n(A)\rightarrow C_{n-1}(X)/C_{n-1}(A)$ yield a chain complex with \emph{relative homology groups} $H_n(X,A):=\mathrm{ker}(\partial_n) / \mathrm{im}(\partial_{n+1})$. Thus, while considering relative homology, one "ignores" the subspace $A$ in some sense. 

There is a connection between relative and absolute homology which greatly increases one's intuition: If $A$ is a nonempty closed subspace of $X$ that is a deformation retract of some neighborhood in $X$, then the pair of spaces $(X,A)$ is called a \emph{good pair}. For good pairs it holds that $H_n(X,A)\cong \tilde{H}_n(X/A)$. (see e.g. \cite{hatcher} for details). 

\subparagraph{Local homology.}
Given a point $q\in X$, the $n$-dimensional \emph{local homology group} of $X$ at $q$ is defined to be the relative homology group $H_n(X,X \setminus q)$. These groups study the local topological structure of $X$ near $q$. Assuming that single points are closed in $X$, excision yields $H_n(X,X\setminus q) \cong H_n(U,U\setminus q)$ for any open neighborhood $U$ of $q$ \cite{hatcher}. Further, notice that if $X$ is a $m$-dimensional manifold then the local homology groups are trivial for $n\neq m$ \cite{dold}. See \cite{dfw-dimdet} for an algorithm to detect dimensions of manifolds sampled by point clouds. 

\subparagraph{Persistent local homology modules.}
To account for the influence of noise, a persistent version of local homology was introduced in \cite{bceh-lhfromstratified}. This concept was adapted by Skraba and Wang \cite{sw-approx} where two types of persistence modules got studied. Before giving the definitions, remember some notation: As before, $P\subset\mathbb{R}^d$ is a finite set of points, $L_s=\bigcup_{p\in P}B_s(p)$ and $q$ is the center. For $s<s'$, $r>r'$ and all maps induced by inclusion, a \emph{persistent local homology module} of $r$-type is given by the diagram 
\begin{equation}
	\label{eq:PLH_type_r}
	\cdots \rightarrow H_n(L_s,L_s\setminus B^o_r(q))\rightarrow H_n(L_s,L_s\setminus B^o_{r'}(q)) \rightarrow \cdots,
\end{equation}
one of $s$-type by
\begin{equation}
	\label{eq:PLH_type_s}
	\cdots \rightarrow H_n(L_s\cap B_r(q),L_s\cap \partial B_r(q))\rightarrow H_n(L_{s'}\cap B_r(q),L_{s'} \cap \partial B_r(q) ).
\end{equation}

Two scenarios are studied here: in case of $r$-type modules one fixes the offsets $L_s$ and varies the scope of locality, whereas for $s$-type modules the shape gets thickened inside a fixed radius ball. 

A connection to classical local homology groups is given by the following direct limits~\cite{b-phd}. $
H_n(L_s,L_s\setminus q)=\lim_{r\rightarrow 0}H_n(L_s,L_s\setminus B^o_r(q))=\lim_{r\rightarrow 0}H_n(L_s\cap B_r(q),L_s\cap \partial B_r(q)).
$

\end{appendix}

\end{document}